\documentclass[11pt, a4paper]{article}

\usepackage[utf8]{inputenc}
\usepackage{amsmath}
\usepackage{amsthm}
\usepackage{mathtools}
\usepackage{braket}
\usepackage{amssymb}
\usepackage{geometry}
\usepackage{csquotes}
\usepackage{fancyhdr}
\usepackage{comment}
\usepackage[toc,page]{appendix}
\usepackage[english]{babel}
\usepackage{hyperref}
\usepackage[export]{adjustbox}
\usepackage{caption}
\usepackage[usestackEOL]{stackengine} %used to split title up
\usepackage{makecell}
\usepackage[boxsize=0.5em,aligntableaux=center]{ytableau}
\usepackage[backend=biber,sorting=none,doi=false,isbn=false,url=false,style=ieee,citestyle=numeric-comp]{biblatex} %citestyle provides the mutliple cite contraction

\allowdisplaybreaks %Allows the align equations to be distributed across pages
\numberwithin{equation}{section}

\usepackage[textsize=tiny,color=green]{todonotes} %enter disable into square brackets to remove todo notes

\newtheoremstyle{named}{}{}{\itshape}{}{\bfseries}{.}{.5em}{\thmnote{#3 }#1}
\theoremstyle{named}
\newtheorem*{namedtheorem}{Lemma}

\def\title#1{\centerline{\LARGE\bf\Longstack{#1}}\vskip .5em}

\def\bea{\begin{eqnarray}}
\def\eea{\end{eqnarray}}

\def\cO{ \mathcal{O} }
\def\cA{ \mathcal{A} }
\def\cB{ \mathcal{B} }
\def\cH{ \mathcal{H} }
\def\cI{ \mathcal{I} }

\DeclareMathOperator{\sech}{sech}

\newcommand{\eqnverty}{\begin{ytableau}
  \ \\
  \ \\
  \ \\
\end{ytableau} }

\newcommand{\eqnhoriy}{\begin{ytableau}
  \ & \none[\cdot] & \
\end{ytableau} }

\newcommand{\eqndiagy}{\begin{ytableau}
  \none[] & \none[] & \ \\
   \none[] & \ &  \none[] \\
  \ &  \none[] &  \none[]
\end{ytableau} }

\newcommand{\eqnrighty}{\begin{ytableau}
   \none[] & \ &  \ \\
  \ &  \ &  \none[]
\end{ytableau} }

\newcommand{\eqnupy}{\begin{ytableau}
  \none[] & \ \\
  \ &  \ \\
   \ &  \none[]
\end{ytableau} }

\begin{document}

\begin{flushright}
QMUL-PH-24-02\\
\end{flushright}

\bigskip
\bigskip

\title{Quantum mechanical bootstrap on the interval: \\
obtaining the exact spectrum }

\bigskip
\bigskip

\centerline{\bf Lewis Sword${}^{\dagger}$ and David Vegh${}^{\dagger \dagger}$}

\bigskip

\begin{center}

\small{
{ \it Centre for Theoretical Physics, Department of Physics and Astronomy \\
Queen Mary University of London, 327 Mile End Road, London E1 4NS, UK}}

\medskip
{\it Email:} ${}^{\dagger}$\texttt{l.sword@qmul.ac.uk}, ${}^{\dagger \dagger}$\texttt{d.vegh@qmul.ac.uk}

\bigskip
\bigskip
\centerline{ \it \today}

\end{center}

\begin{abstract}

We show that for a particular model, the quantum mechanical bootstrap is capable of finding exact results. We consider a solvable system with Hamiltonian $H=SZ(1-Z)S$, where $Z$ and $S$ satisfy canonical commutation relations. While this model may appear unusual, using an appropriate coordinate transformation, the Schr{\"o}dinger equation can be cast into a standard form with a P{\"o}schl--Teller-type potential. Since the system is defined on an interval, it is well-known that $S$ is not self-adjoint. Nevertheless, the bootstrap method can still be implemented, producing an infinite set of positivity constraints. Using a certain operator ordering, the energy eigenvalues are only constrained into bands. With an alternative ordering, however, we find that a finite number of constraints is sufficient to fix the low-lying energy levels exactly.
\end{abstract}

\tableofcontents

%%%%%%%%%%%%%%%%%%%%%%%%%%%%%%%%%%%%%%%%%%%%%%%%%%%%%%%%%%%%%%%%%%%%SECTION%%%%%%%%%%%%%%%%%%%%%%%%%%%%%%%%%%%%%%%%%%%%%%%%%%%%%%%%%%%%%%%

\section{Introduction}
The quantum mechanical bootstrap provides a method to numerically approximate the expectation values of a given system \cite{Han:2020bkb}. Initially proposed as a way to solve random matrix models \cite{Lin:2020mme}, numerous systems have now been analysed with the bootstrap procedure including: the quantum anharmonic oscillator which produces rapidly converging energy vs. position expectation value ``islands" with increased bootstrap matrix size \cite{Han:2020bkb,Bhattacharya:2021btd,Aikawa:2021qbl}, the Mathieu problem with band behaviour \cite{Berenstein:2021loy}, PT-symmetric systems \cite{Khan:2022uyz} and even models with exponentiated canonical operators as found in Calabi-Yau discussions \cite{Du:2021hfw}. For a selection of alternative systems and further explorations within the literature, see \cite{Berenstein:2021dyf,Nancarrow:2022wdr,Berenstein:2022ygg,Berenstein:2022unr,Hu:2022keu,Aikawa:2021eai}.

In this paper, we use the bootstrap construction to constrain the energy eigenvalues of a system defined on the interval. Remarkably, the constraints are sufficiently strong to fix these energy eigenvalues exactly. We begin by discussing the Hamiltonian, its analytic solutions and the boundary conditions in \S \ref{s:model}. The bootstrap is then reviewed in \S \ref{s:bootstrap}. Here we describe the method to find recursion relations and assess the associated anomalies.  We then detail the composition of the bootstrap matrices that are built from such relations. In \S \ref{s:numerical_results} we provide the numerical results and figures depicting both the band-like and exact behaviour. Finally, \S \ref{s:conclusion} provides a summary of the findings and suggests future directions of investigation. The appendices provide additional technical details of the calculations.

%%%%%%%%%%%%%%%%%%%%%%%%%%%%%%%%%%%%%%%%%%%%%%%%%%%%%%%%%%%%%%%%%%%%SECTION%%%%%%%%%%%%%%%%%%%%%%%%%%%%%%%%%%%%%%%%%%%%%%%%%%%%%%%%%%%%%%%

\section{The model}
\label{s:model}

\subsection{The Hilbert space}
\label{ss:the_hilbert_space}

We take the Hilbert space $\cH$ to be the space of square integrable functions over the interval $z \in [0,1]$. The inner product is defined on $\cH$ as
\begin{equation}
\langle\phi| \psi \rangle = \int_{0}^{1}  \phi(z)^{*} \psi(z) dz\,.
\end{equation}
Note that this also defines a product on a larger function space consisting of potentially non-square integrable functions. The norm on $\cH$ is defined by $||\psi|| = \sqrt{\braket{\psi | \psi}}$.

Let us consider a densely defined linear operator $A: D(A) \rightarrow \cH$. The domain and the action of $A^{\dagger}$ are defined by 
\begin{enumerate}
\item $D(A^{\dagger}):= \{\phi \in \cH | \exists \eta \in \cH : \forall \alpha \in D(A): \langle \phi | A \alpha \rangle = \langle \eta | \alpha \rangle \} $
\item $  A^{\dagger} \phi := \eta $
\end{enumerate}
Self-adjoint operators must satisfy two properties:
\begin{enumerate}
\item Symmetricity: $\langle A \phi| \psi \rangle -\langle \phi| A \psi \rangle =0 \,,$
\item Equality of operator domains:
$D(A) = D(A^{\dagger})\,\,,$
\end{enumerate}
for arbitrary wavefunctions $\phi \in D(A^{\dagger}), \psi \in D(A)$. In the following we consider the canonical operators $S$ and $Z$. In the $z$-basis, $S = i \hbar \partial_{z}$ (note the sign choice) and $Z = z \cdot$ satisfy the commutation relation $[S, Z] = i\hbar$, and henceforth we take $\hbar=1$. It is well-known that the operator $S$ is not self-adjoint on the interval\footnote{There does exist a self-adjoint extension of $S$ on the interval, see for example \cite{Juric:2021psr} where the wavefunctions are restricted to have boundary values identified up to a parameter dependent phase. Further information about extensions on the interval can be found in \cite{Al-Hashimi:2021tkf,PhysRevResearch.3.L042008}.}.

\subsection{The Hamiltonian}
We consider the Hamiltonian \cite{Vegh:2023snc}
\begin{equation}
\label{eqn:model_ham_in_s_z}
H = SZ(1-Z)S\,,
\end{equation}
In the $z$-basis, the time-independent Schr{\"o}dinger equation becomes
\begin{equation}
\label{eqn:model_ham_equation}
H \psi(z) = E \psi(z) \,, \quad z(1-z)\psi''(z) + (1-2z)\psi'(z) +E \psi(z) = 0\,,
\end{equation}
where $E$ is the eigenvalue of the system for the particular energy eigenfunction $\psi(z)$. We want $H$ to be self-adjoint and therefore we consider it on the dense domain
\begin{equation}
\label{eqn:model_dom_h}
D(H) := \{ \psi, H \psi \in \cH | \lim_{z \to 0, 1} z(1-z) \partial_{z}\psi(z) = 0 \} \,.
\end{equation}
$H$ must be symmetric, which means for any $\phi \in \cH$,  $H\phi \in \cH$ and $\psi \in D(H) \subset \cH$
\begin{equation}
\label{eqn:model_symmetricity_eqn}
\langle \phi| H \psi \rangle - \langle H \phi| \psi \rangle  =0 \,.
\end{equation}Note that $\phi$ is not necessarily in $D(H)$, therefore what is meant by $H\phi$ is that the differential operator $H$ acts on $\phi(z)$. Since $|| H \phi|| < \infty$, $\phi$ must at most be logarithmically divergent and subleading terms are either constant or vanish faster than $z^{1/2}$ or $(1-z)^{1/2}$:
\begin{align}
\label{eqn:model_norm_cond_on_phi}
&\phi(z)  = c_{0} + c_{0}'\log{\left(\frac{z}{1-z}\right)} + \cO \left(z^{\frac{1}{2}}\right) \,, \text{ as } z \to 0
\\
& \phi(z)  = c_{1} + c_{1}' \log{\left(\frac{z}{1-z}\right)} + \cO \left((1-z)^{\frac{1}{2}}\right) \,, \text{ as } z \to 1
\end{align}
where $c_{0}, c_{0}',c_{1}$ and $c_{1}'$ are constants. Since $\psi \in D(H)$, the condition in \eqref{eqn:model_dom_h} eliminates the logarithmically divergent terms,
\begin{align}
\label{eqn:model_norm_cond_on_psi_z0}
&\psi(z)  = c_{0}'' + \cO \left(z^{\frac{1}{2}}\right) \,, \text{ as } z \to 0
\\
\label{eqn:model_norm_cond_on_psi_z1}
&\psi(z)  = c_{1}'' + \cO \left((1-z)^{\frac{1}{2}}\right) \,, \text{ as } z \to 1
\end{align}
where $c_{0}''$ and $c_{1}''$ are constants. To perform the symmetricity calculation of \eqref{eqn:model_symmetricity_eqn}, we focus on its first term and integrate by parts
\begin{align}
\label{eqn:model_h_self_adj_bound_term}
\braket{\phi | H \psi}
&= -\int^{1}_{0} \phi(z)^{*} \partial_{z}\left[ z(1-z)\partial_{z} \psi(z) \right]  dz
\\
\label{eqn:model_terms_to_inf_gbc}
&= -\phi(z)^{*} z(1-z) \partial_{z} \psi(z)\bigg|_{z \to 0}^{z \to 1} + \int_{0}^{1} (\partial_{z} \phi(z)^{*}) z(1-z) \partial_{z}\psi(z) dz
\\
\begin{split}
&= -\phi(z)^{*} z(1-z) \partial_{z} \psi(z)\bigg|_{z \to 0}^{z \to 1} + \Bigg( \psi(z)  z(1-z)  \partial_{z} \phi(z)^{*}\bigg|_{z \to 0}^{z \to 1}
\\
& \qquad \qquad  \qquad \qquad  \qquad \qquad  \qquad \qquad \quad - \int_{0}^{1} \partial_{z} \left[z(1-z) \partial_{z}\phi(z)^{*} \right] \psi(z) dz \Bigg)
\end{split}
\\
\label{eqn:model_bts_hphi_psi}
&=  -\phi(z)^{*} z(1-z) \partial_{z} \psi(z)\bigg|_{z \to 0}^{z \to 1} + \psi(z)  z(1-z)  \partial_{z} \phi(z)^{*} \bigg|_{z \to 0}^{z \to 1} + \langle H\phi| \psi \rangle \,,
\end{align}
Since $\phi^{*}$, $\partial_{z} \psi$ and $\partial_{z} \phi^{*}$ may potentially blow up, these boundary terms must be evaluated using limits $\lim_{z \to 0,1}$. Each of the expressions in \eqref{eqn:model_h_self_adj_bound_term}-\eqref{eqn:model_bts_hphi_psi} must be finite. Since $\psi \in \cH$, then $\langle H\phi| \psi \rangle$ meets this criteria. The first term in \eqref{eqn:model_bts_hphi_psi} always vanishes. The second term is finite owing to the $z(1-z)$ factor, but generally does not vanish. A natural choice of conditions that ensure it does vanish is
\begin{equation}
\label{eqn:model_bcs}
\qquad \lim_{z \to 0, 1} z(1-z) \partial_{z}\phi(z) = 0\,.
\end{equation}
The set of $\phi$'s which satisfy \eqref{eqn:model_bcs} determines $D(H^{\dagger})$. We conclude that $H$ is a self-adjoint operator since it is symmetric and $D(H) = D(H^{\dagger})$.
Although we will use the boundary conditions of \eqref{eqn:model_dom_h} and \eqref{eqn:model_bcs} throughout, we note that more general boundary conditions do exist, see \S \ref{ss:gen_bound_con}.

\subsection{The spectrum}
\label{ss:spectrum}
The Schr{\"o}dinger equation \eqref{eqn:model_ham_equation} has a general analytic solution
\begin{equation}
\label{eqn:model_gen_sol}
\psi_{n}(z) = a_{1} P_{n} \left(  2z-1 \right) + a_{2} Q_{n} \left( 2z-1 \right) \,,
\end{equation}
where $a_{1}$ and $a_{2}$ are constants, $P_{n}$  is the Legendre polynomial, $Q_{n}$ is the Legendre function of the second kind and $n = \frac{1}{2} \left(-1 + \sqrt{1+ 4 E} \right)$. Note that $P_{n}$ and $Q_{n}$ are real.

Evaluating the boundary conditions
\begin{equation}
\label{eqn:model_bc_psi}
\lim_{z \to 0,1} z(1-z) \psi'(z)=0
\end{equation}
using the general solution \eqref{eqn:model_gen_sol} implies that $n=0,1,2 \dots$ and $a_{2}=0$. The non-negative integer $n$ provides the quantisation condition for the energy eigenvalues
\begin{equation}
\label{eqn:model_quant_cond}
\braket{\psi_{n}| H| \psi_{n}} \equiv E_{n} = n(n+1) \,,
\end{equation}
where
\begin{equation}
\label{eqn:model_analytic_sol}
\psi_{n}(z) = \sqrt{2n +1} P_{n} \left(2z-1 \right) \,,
\end{equation}
and constant $a_{1}$ has been fixed by normalisation condition
$\langle \psi_{n} | \psi_{n} \rangle = 1$. Having the analytic solution at hand provides a useful diagnostic tool for the bootstrap.

%%%%%%%%%%%%%%%%%%%%%%%%%%%%%%%%%%%%%%%%%%%%%%%%%%%%%%%%%%%%%%%%%%%%SUBSECTION%%%%%%%%%%%%%%%%%%%%%%%%%%%%%%%%%%%%%%%%%%%%%%%%%%%%%%%%%%%%%%%

\subsection{P{\"o}schl-Teller coordinates}
\label{ss:po_tell_coords}
To understand where the P{\"o}schl-Teller potential is manifest in our system, we can choose alternative coordinate
\begin{equation}
\label{eqn:model_z_to_p_coord}
z = \frac{e^{p}}{1+ e^{p}}\,,
\end{equation}
such that the Schr{\"o}dinger equation becomes
\begin{equation}
\label{eqn:model_diff_eqn_in_p}
\psi''(p) + \frac{E}{4 \cosh^{2}\left(\frac{p}{2} \right)}\psi(p) = 0 \,.
\end{equation}	
This is related to the work of \cite{Khan:2022uyz} where they studied equation
\begin{equation}
\label{eqn:model_khan_diff_eqn}
\psi''(p) + \frac{\lambda(\lambda+1)}{4 \cosh^{2}\left(\frac{p}{2}\right)} \psi(p) = -\frac{E_{K}}{2} \psi(p)
\end{equation}
and applied the bootstrap to identify the eigenvalues $E_{K}$. Notably, by taking $E_{K} = 0$ and $\lambda(\lambda+1)=E$ the two equations coincide\footnote{There appears to be a typo in the original recursion relation proposed in \cite{Khan:2022uyz}. We find the relation to be  
\begin{align*}
 \tiny &\left(4 E_{K} t + \frac{t^3}{2} \right)\langle \sech^{t}(x)\rangle + \left( -4 E_{K}(t+1) + 2(t+1)\lambda(1+\lambda) -t^3 -3t^2 -4t -2 \right)\langle \sech^{t+2}(x)\rangle 
 \\
 &\qquad \qquad \qquad \qquad \qquad \qquad \qquad \qquad \qquad   + \left( \frac{t^3}{2} + 3t^2 + \frac{11}{2}t +3 -2(t+2)\lambda(\lambda+1)\right) \langle \sech^{t+4}(x)\rangle = 0 \,,
\end{align*}
where $x= p/2$. Note that upon scaling/redefining $E_{K} \to \frac{1}{2} E_{K}$ and swapping the $\langle \sech^{t+2}(x)\rangle $ and $\langle \sech^{t+4}(x)\rangle $ coefficients in the above expression, this is equal to equations (2.5) and (4.25) of \cite{Khan:2022uyz}.}. Under this coordinate transformation, the boundary condition of equation \eqref{eqn:model_bc_psi} becomes
\begin{equation}
\label{eqn:model_neumann_bc}
\lim_{p\to -\infty, \infty} \psi'(p) = 0
\end{equation}
which is a Neumann-type boundary condition. Another transformation is achieved by setting
\begin{equation}
\label{eqn:model_z_to_theta}
z = \frac{1}{2}\left(1+\sin(\theta)\right) \,.
\end{equation}

The Schr{\"o}dinger equation \eqref{eqn:model_ham_equation} then becomes
\begin{equation}
\label{eqn:model_schro_in_theta}
\psi''(\theta) -\tan(\theta) \psi'(\theta) + E \psi(\theta) = 0 \,,
\end{equation}
which, upon redefining the wavefunction and the energy to
\begin{equation}
\psi(\theta)= \frac{1}{\sqrt{\cos{\theta}}}\tilde{\psi}(\theta) \,, \qquad E = \tilde{E}-\frac{1}{4} \,,
\end{equation}
may be rewritten as
\begin{equation}
-\tilde{\psi}''(\theta) -\frac{1}{4 \cos^{2}(\theta)}\tilde{\psi}(\theta) = \tilde{E} \tilde{\psi}(\theta) \,.
\end{equation}

%%%%%%%%%%%%%%%%%%%%%%%%%%%%%%%%%%%%%%%%%%%%%%%%%%%%%%%%%%%%%%%%%%%%SUBSECTION%%%%%%%%%%%%%%%%%%%%%%%%%%%%%%%%%%%%%%%%%%%%%%%%%%%%%%%%%%%%%%%

\subsection{General boundary conditions}
\label{ss:gen_bound_con}

The expectation values $\braket{\phi|H \psi}$ and $\braket{H\phi| \psi}$ in \eqref{eqn:model_symmetricity_eqn} remain well-defined when both $\psi$ and $\phi$ have logarithmic divergences.  However, upon integrating by parts the boundary term and resultant integral in \eqref{eqn:model_terms_to_inf_gbc} will become infinite as $z \to 0,1$. If we introduce a small $\epsilon$ cut-off and change the Hilbert space to the space of square integrable functions over the interval $z \in [\epsilon, 1-\epsilon]$, then the expressions stay well-defined.  Using wavefunction ans{\"a}tze
\begin{align}
\label{eqn:model_genbc_norm_cond_on_phi}
&\phi(z)  = d_{0} + d_{0}' \log{\left(\frac{z}{1-z}\right)} + \cO \left(z^{\frac{1}{2}}\right) \,, \text{ as } z \to \epsilon
\\
& \phi(z)  = d_{1} + d_{1}' \log{\left(\frac{z}{1-z}\right)} + \cO \left((1-z)^{\frac{1}{2}}\right) \,, \text{ as } z \to 1-\epsilon
\end{align}
where $d_{0}, d_{0}',d_{1}$ and $d_{1}'$ are constants and
\begin{align}
\label{eqn:model_genbc_norm_cond_on_psi_z0}
&\psi(z)  = d_{0}'' + d_{0}''' \log{\left(\frac{z}{1-z}\right)} + \cO \left(z^{\frac{1}{2}}\right) \,, \text{ as } z \to \epsilon
\\
\label{eqn:model_genbc_norm_cond_on_psi_z1}
&\psi(z)  = d_{1}'' + d_{1}''' \log{\left(\frac{z}{1-z}\right)}  + \cO \left((1-z)^{\frac{1}{2}}\right) \,, \text{ as } z \to 1-\epsilon
\end{align}
where $d_{0}'', d_{0}''', d_{1}''$ and $d_{1}'''$ are constants, the general boundary conditions can be written as
\begin{align}
\label{eqn:model_gen_bound_z0}
&\left[ \alpha_{0} z(1-z) \psi'(z) + \psi(z)-z(1-z)\log{\left(\frac{z}{1-z} \right)}\psi'(z) \right]\bigg|_{z = \epsilon} = 0 \,,
\\
\label{eqn:model_gen_bound_z1}
&\left[\alpha_{1} z(1-z) \psi'(z) + \psi(z)-z(1-z)\log{\left(\frac{z}{1-z} \right)}\psi'(z) \right]\bigg|_{z = 1- \epsilon} = 0 \,,
\end{align}
where $\alpha_{0,1}$ are constants and these same boundary conditions are also applied to $\phi(z)$. By plugging the functions \eqref{eqn:model_genbc_norm_cond_on_phi}-\eqref{eqn:model_genbc_norm_cond_on_psi_z1} into \eqref{eqn:model_gen_bound_z0} and \eqref{eqn:model_gen_bound_z1} respectively, this enforces 
\begin{equation}
d_{0}  +\alpha_{0} d_{0}' = d_{1} + \alpha_{1} d_{1}' = d_{0}''  +\alpha_{0} d_{0}''' = d_{1}'' + \alpha_{1} d_{1}''' = 0 \,,
\end{equation}
in the $\epsilon \to 0$ limit. Providing $\alpha_{0} \in \mathbb{R}$ and $\alpha_{1} \in \mathbb{R}$, these equations mean the boundary terms are equal in the symmetricity calculation \eqref{eqn:model_bts_hphi_psi} and will cancel, thus $H$ is self-adjoint. Note, in the $p$ coordinates of \S\ref{ss:po_tell_coords}, the boundary conditions become
\begin{align}
& \psi(p) + (\alpha_{0} - p ) \psi'(p)\big|_{p = a_{-}} = 0
\\
&\psi(p) + (\alpha_{1} - p) \psi'(p)\big|_{p = a_{+}} = 0
\end{align}
where $a_{-} = \log{\frac{\epsilon}{1-\epsilon}}$ and $a_{+} = \log{\frac{1-\epsilon}{\epsilon}}$. These are similar to Robin boundary conditions.

%%%%%%%%%%%%%%%%%%%%%%%%%%%%%%%%%%%%%%%%%%%%%%%%%%%%%%%%%%%%%%%%%%%%SECTION%%%%%%%%%%%%%%%%%%%%%%%%%%%%%%%%%%%%%%%%%%%%%%%%%%%%%%%%%%%%%%%

\section{Bootstrap}
\label{s:bootstrap}

In the simplest cases, the bootstrap begins with a Hamiltonian $H$ and a set of operators, $\{\cO_{i} \}$, where the type and the number of operators required, depends on the model. These are used to produce energy and commutation equations
\begin{equation}
\langle H \cO_{i} \rangle - E \langle \cO_{i} \rangle = 0 \qquad \text{and} \qquad \langle [H, \cO_{i}] \rangle =0 \,.
\end{equation}
Note that $\langle \cO_{i} \rangle \equiv \bra{\psi_{E}} \cO_{i} \ket{\psi_{E}}$, where $\ket{\psi_{E}}$ is a specific energy eigenstate. After using the commutation rules between the various required $\cO_{i}$ and $H$, it is potentially possible to form closed recursion relations between expectation values, based on the contents of the $\cO_{i}$. A subset of these expectation values therefore generate the remaining moments and together with the energy $E$, they form a set of \textit{initial data} for the bootstrap. In quantum mechanics, a general operator $\cO$  must adhere to the positivity constraint
\begin{equation}
\label{eqn:boot_bootstrap_constraint}
\langle  \cO^{\dagger} \cO \rangle \geq 0 \,.
\end{equation}
The authors of \cite{Han:2020bkb} showed that by writing the constraint as a positive semi-definiteness condition on a Hermitian matrix
\begin{equation}
\label{eqn:boot_boot_mat_intro}
\cB \succeq 0 \,,
\end{equation}
the expectation values of a system can be bounded. We call $\cB$ the bootstrap matrix and note that $\cB \succeq 0$ is equivalent to $\cB$'s eigenvalues being non-negative. The bootstrap operator $\cO$ used to create $\cB$ is often based on one of the aforementioned $\cO_{i}$, however it can be a general combination from the set, as long as the recursion relation allows one to obtain all the $\cB$ matrix elements. Therefore, via these relations, we populate this matrix with different choices of initial data and test the positivity. If the result is negative, then the trialled initial data is ruled out. By reiterating the process for different initial data sets we aim to carve out regions in the initial data space that satisfy the constraint, indicating where the actual, physical values exist. Increasing the matrix size by extending\footnote{``Extending" refers to increasing the range of $i$ when writing the operator as a linear combination $\cO = \sum_{i} c_{i} \tilde{\cO}_{i}$, where $\tilde{\cO}_{i}$ are selected operators. This translates to an increase in size of the $\cB$ matrix. In theory, one can always increase the range of $i$ which is why the set of constraints are potentially infinite.} $\cO$ can lead to convergence in some cases \cite{Du:2021hfw} and as such, the size of the matrix provides a parameter labelled ``$K$", for a $K \times K$ bootstrap matrix. The remainder of this section provides explicit details on the recursion relations, as well as the different types of matrix construction we choose to explore.

%%%%%%%%%%%%%%%%%%%%%%%%%%%%%%%%%%%%%%%%%%%%%%%%%%%%%%%%%%%%%%%%%%%%%%%%%%%%%%%%%%%%%%%%%SUBSECTION%%%%%%%%%%%%%%%%%%%%%%%%%%%%%%%%%%%%%%%%%%%%%%%%%%%%%%%%%%%%%%%%%%%%%%%%%%%%%%%%%%%%%%%%%
\subsection{Recursion relations}
\label{ss:rec_rels}
We define general 2d moments
\begin{equation}
f_{\sigma,\zeta} = \langle S^\sigma Z^\zeta \rangle\,,
\end{equation}
as well as 1d moments
\begin{equation}
f_{\zeta} \equiv f_{0,\zeta}  = \langle Z^\zeta \rangle \,,  \qquad \tilde{f}_{\sigma} \equiv f_{\sigma, 0} = \langle S^{\sigma} \rangle \,,
\end{equation}
and, as in \cite{Han:2020bkb}, find recursion relations between them. Here and in the rest of the paper, $\sigma$ and $\zeta$ are integers. We begin by inserting operator $\cO_{e}(S,Z) = S^\sigma Z^\zeta$ into the two possible forms of energy equation and upon completing the various commutations, we arrive at
{\small \begin{align}
\label{eqn:boot_2d_ho_eqn}
\begin{split}
\langle H \cO_{e} \rangle - E \langle \cO_{e} \rangle &=
\left(\sigma(\sigma+1) - E\right)f_{\sigma,\zeta} - i (\sigma+1)f_{\sigma+1,\zeta} + 2i (\sigma+1)f_{\sigma+1,\zeta+1} + f_{\sigma+2,\zeta+1}-f_{\sigma +2, \zeta+2} = 0 \,.
\end{split}
\end{align}}
{\small \begin{equation}
\label{eqn:boot_2d_oh_eqn}
\langle \cO_{e}  H  \rangle - E \langle \cO_{e} \rangle =
-\zeta^2 f_{\sigma,\zeta-1} + (\zeta^2 +\zeta -E)f_{\sigma, \zeta}-i(2\zeta +1)f_{\sigma+1, \zeta} + 2i(\zeta+1) f_{\sigma +1, \zeta+1} + f_{\sigma +2, \zeta+1} - f_{\sigma+2,\zeta+2} =0 \,.
\end{equation}}With these, we can construct $\langle [H,\cO_{e}]\rangle = 0$ and in turn find equations that relate the moments in this 2-dimensional $(\sigma, \zeta)$ space. We find
\begin{equation}
\label{eqn:boot_up_eqn}
\text{\eqnupy } : \quad -\zeta^2 f_{\sigma,\zeta-1} - (\sigma -\zeta)(1+\sigma +\zeta)f_{\sigma,\zeta} + i(\sigma -2 \zeta)f_{\sigma+1,\zeta} + 2i (\zeta-\sigma)f_{\sigma+1,\zeta+1} = 0 
\end{equation}
\begin{align}
\label{eqn:boot_right_eqn}
\begin{split}
\text{\eqnrighty } : &\quad -2(\sigma -\zeta)(E-\sigma(\sigma+1))f_{\sigma,\zeta} - i(1+ \zeta^2 - 2\zeta \sigma +2\sigma(\sigma+1))f_{\sigma+1,\zeta} 
\\
& \qquad \qquad \qquad \qquad + i(\sigma-\zeta)(1-\zeta +3\sigma)f_{\sigma+1,\zeta+1} + (\sigma+1)f_{\sigma+2,\zeta+1} = 0 
\end{split}
\end{align}
These two equations relate general 2d moments in the $(\sigma, \zeta)$ space. We choose to denote these recursion relations by block symbols based on how their constituent moments appear on  the $(\sigma, \zeta)$ plane. \eqnupy is simply the commutator $\langle [H, \cO_{e}] \rangle = 0$. \eqnrighty is generated by shifting indices in \eqnupy by $(\sigma, \zeta) \to (\sigma +1, \zeta +1)$, solving for $f_{\sigma +2, \zeta+2}$ and substituting this into \eqref{eqn:boot_2d_ho_eqn}. 

From equations \eqref{eqn:boot_up_eqn} and \eqref{eqn:boot_right_eqn} we can find three recursion relations which are one-dimensional i.e. each depends on a single index,
\begin{equation}
\label{eqn:boot_1d_rec_in_z}
\text{ \eqnverty} : \quad(\zeta-1)^3 f_{\zeta-2}+(2 \zeta-1) \left(2 E - \zeta^2+ \zeta \right) f_{\zeta-1}+ \zeta \left( \zeta^2 -4 E -1\right) f_{\zeta}=0
\end{equation}
\begin{equation}
\label{eqn:boot_1d_diag}
\text{ \eqndiagy } : \quad (\zeta^2 + \zeta -E )f_{\zeta,\zeta} + 2 i (\zeta+1) f_{\zeta+1,\zeta+1} - f_{\zeta+2,\zeta +2}=0
\end{equation}
\begin{equation}
\label{eqn:boot_1d_in_s_hori}
\text{\eqnhoriy } : \quad \sigma(4 E + 1 - \sigma^2)\tilde{f}_{\sigma} - (\sigma+1) \tilde{f}_{\sigma+2}=0
\end{equation}
Relation \eqnverty is found\footnote{For an alternative derivation, see Appendix \ref{app:drr}.} by forming the following linear combination of equations \eqref{eqn:boot_up_eqn} and \eqref{eqn:boot_right_eqn}
\begin{equation}
E \cdot \eqnupy \bigg|_{(\sigma=-1, \zeta)} -\frac{\zeta}{2} \cdot \eqnrighty \bigg|_{(\sigma=-1, \zeta)} + \frac{\zeta}{2}  \cdot \eqnrighty \bigg|_{(\sigma=-1,\zeta-1)} \,,
\end{equation}
then shifting the result by $\zeta \to \zeta-1$.
\eqndiagy is found by setting $\sigma = \zeta$ in \eqnupy , followed by shifting the indices $\zeta \to \zeta+1$ , solving for $f_{\zeta+2,\zeta+1}$ and then substituting this into equation \eqref{eqn:boot_2d_ho_eqn}, which also sets $\sigma = \zeta$. To obtain \eqnhoriy , we first solve \eqnupy at $\zeta=0$ and $\zeta=1$ for $f_{\sigma+1,1}$ and $f_{\sigma+1,2}$, then shift these expectation values by $\sigma\to \sigma+1$ to find $f_{\sigma+2,1}$ and $f_{\sigma+2,2}$. We then substitute these into equation \eqref{eqn:boot_2d_oh_eqn} at $\zeta =0$ and solve for $f_{\sigma +1,1}$. Finally, we insert $f_{\sigma+2,2}$, $f_{\sigma+2,1}$ and $f_{\sigma +1,1}$ into equation \eqref{eqn:boot_2d_ho_eqn}, evaluated at $\zeta=0$.

By setting $\zeta = 1$ in \eqnverty of \eqref{eqn:boot_1d_rec_in_z}, we find relation $f_{1}= \frac{1}{2}f_{0}$. Since we normalise using $f_{0,0}=f_{0} =\tilde{f}_{0}=1$, all higher order $f_{\zeta}$ moments are therefore defined by a single piece of initial data\footnote{One dimensional initial data spaces/search spaces have been encountered in previous studies: the Coulomb potential model in \cite{Berenstein:2021dyf} for instance.}: $\{E\}$. The other 1d recursion relations found, \eqndiagy and \eqnhoriy , have two dimensional search/initial data spaces: $\{E,f_{1,1}\}$ and $\{E,f_{1,0} \}$ respectively. Any positive quadrant ($\sigma \geq 0, \zeta \geq 0$), 2d moment $f_{\sigma, \zeta}$ can be generated using a selection of the five total recursion relations and will generally depend on the initial data: $\{E, f_{1,0}, f_{1,1} \}$.

As a preliminary check, we evaluate the moments $f_{2}$ using \eqnverty and compare this to the direct integration method. From the recursion relation, one obtains
\begin{equation}
f_{2} = \langle Z^2 \rangle = \frac{3E -2}{2(4E - 3)}f_{0} \,.
\end{equation}
Alternatively, using the analytic solution $\psi_{n}$ and calculating $\braket{\psi_{n}|Z^2| \psi_{n}}$ via integration, we achieve the same result upon applying $E= n(n+1)$.

\subsection{Restrictions on \texorpdfstring{$(\sigma,\zeta)$}{}}
\label{ss:octant_choice}

We choose to restrict the 2d $(\sigma, \zeta)$ space we explore, considering only a subset of $f_{\sigma, \zeta}$ moments. This ensures that the expectation values we use are finite. The three ``negative" quadrants in the $(\sigma, \zeta)$ space\footnote{Those quadrants defined by $( \sigma < 0, \zeta \geq 0)$, $(\sigma < 0, \zeta < 0 )$ and  $(\sigma \geq 0, \zeta < 0 )$.}, introduce operators with negative exponents, and when the associated operators have a well defined action on an eigenstate, they can produce non-convergent expectation values from analytic checks.
As a simple example, if we take $\sigma = 0, \zeta = -1$, then $f_{\sigma, \zeta}$ evaluated on some energy eigenstate $\psi$ is
\begin{equation}
f_{0, -1} = \braket{\psi| Z^{-1} | \psi} = \int_{0}^{1} \frac{1}{z} |\psi(z)|^2 dz = \infty \,,
\end{equation}
owing to the polynomial form of the eigenstate solutions $\psi$. To avoid such issues, we restrict the range to the positive quadrant, $(\sigma \geq 0, \zeta \geq 0)$. A recursion relation can be plotted on a lattice of $\zeta$ vs. $\sigma$ points, by placing a square at the pair of integers $(\sigma, \zeta)$ corresponding to every moment $f_{\sigma,\zeta}$ present in said relation. Figure \ref{fig:octant_with_eqns} demonstrates these equation plots as well as the forbidden (red) and permitted (dark and light green) regions for moments $f_{\sigma, \zeta}$ in the $(\sigma, \zeta)$ plane.
\begin{figure}[t!]
\begin{center}
\includegraphics[width=5.5cm, height=5.5cm]{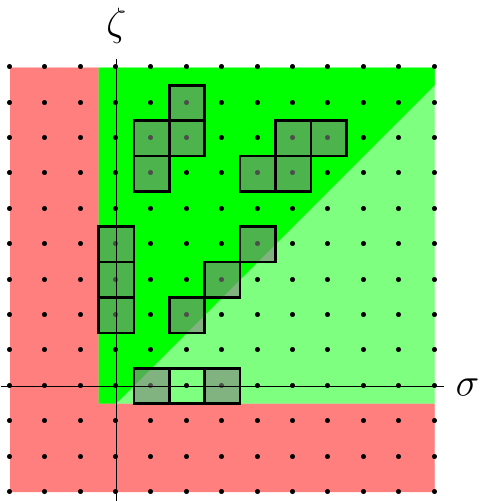}
\end{center}
\caption{A section of the ($\sigma, \zeta$) plane, displaying the lattice representations of the five recursion relations, based on the moments $f_{\sigma, \zeta} = \langle S^{\sigma} Z^{\zeta} \rangle$ they contain. The green regions ($\sigma \geq 0, \zeta \geq 0)$ cover $(\sigma, \zeta)$ pairs which define permitted $f_{\sigma, \zeta}$, while the red region represents the moments we do not explore. The darker green region ($\zeta \geq \sigma \geq 0 $) contains the $f_{\sigma, \zeta}$ that will be used in specific bootstrap matrices.  }
\label{fig:octant_with_eqns}
\end{figure}

%%%%%%%%%%%%%%%%%%%%%%%%%%%%%%%%%%%%%%%%%%%%%%%%%%%%%%%%%%%%%%%%%%%%%%%%%%%%%%%%%%%%%%%%%SUBSECTION%%%%%%%%%%%%%%%%%%%%%%%%%%%%%%%%%%%%%%%%%%%%%%%%%%%%%%%%%%%%%%%%%%%%%%%%%%%%%%%%%%%%%%%%%

\subsection{Anomalies}
\label{ss:anomalies}

Corrections to the recursion relations may occur due to anomalies \cite{Juric:2021psr,Esteve:2002nt} which can appear upon careful consideration of operator domains. Let us take an operator, $A$, for which $[H,A] = 0$. As seen from the definition of self-adjointness in \S\ref{ss:the_hilbert_space}, it is important to consider the domains of such operators, $D(A)$ and $D(H)$. Indeed, given that $\psi \in D(H)$, if equation
\begin{equation}
[H,A] \psi = (HA - A H) \psi = H(A \psi) - A (H\psi) \,,
\end{equation}
is to make sense, we must have $A \psi \in D(H)$ and $H\psi \in D(A)$. If $A \psi \notin D(H)$ then we say that $A$ does not leave the domain of $H$ invariant, the symmetry generated by $A$ is broken and an anomaly appears. The exact form of this commutator-based anomaly, expressed in terms of $A$ and $H$, can be found by deriving the Heisenberg equation in the Hamiltonian formalism
\begin{align}
\label{eqn:anom_heis_eqn}
\frac{d}{dt}\braket{\psi(t) | A \psi(t)} &= \braket{\frac{\partial \psi(t)}{\partial t} |A \psi(t)} + \braket{\psi(t)|\frac{\partial A}{\partial t} \psi(t)} + \braket{\psi(t)|A \frac{\partial \psi(t)}{\partial t}}
\\
&=\braket{-i H \psi(t) |A \psi(t)} + \braket{\frac{\partial A}{\partial t}}_{\psi} + \braket{\psi(t)|(-i) A H\psi(t)}
\\
\label{eqn:anom_heis_eqn_fin}
&= \braket{\frac{\partial A}{\partial t}}_{\psi} + i \braket{ H \psi(t) |A \psi(t)} - i \braket{\psi(t)| A H\psi(t)}\,.
\end{align}
Here $\psi(t)$ belongs to Hilbert space $\cH$, the Hamiltonian $H$ is self-adjoint with domain $D(H)$ and $A$ is a general operator. If we assume that $\psi \in D(H) \cap D(A)$ then $\braket{ H \psi(t) |A \psi(t)}$ is well-defined. However, this assumption does not cover the third term in \eqref{eqn:anom_heis_eqn_fin}, since it requires $H\psi \in D(A)$. If $[A,H]$ is well-defined in all of $\cH$, then we may make replacement $AH = [A,H] + HA$. Substituting this in, and using that $\braket{ H \psi(t) |A \psi(t)} \equiv \braket{ \psi(t) | H^{\dagger} A \psi(t)}$, we find
\begin{align} 	
\frac{d}{dt}\braket{\psi(t) | A \psi(t)} = \braket{\frac{\partial A}{\partial t}}_{\psi(t)} + i\braket{\psi(t) |[H,A] \psi(t)} + i \braket{\psi(t)|(H^{\dagger} - H) A\psi(t)} \,.
\end{align}
Therefore, we see that the Heisenberg equation receives a correction owing to the subtleties of the operator domains. It indicates that when computing the expectation value of a commutator, there is an additional contribution
\begin{equation}
\label{eqn:anom_corr_rec}
\braket{\psi |[H,A]  \psi}_{\text{total}} = \braket{\psi |[H,A] \psi}_{\text{reg}} + \cA_{A} \,.
\end{equation}
Here, $\braket{\psi | [H,A] \psi}_{\text{reg}}$ is the standard  algebraic commutator expectation value, where again, $[H,A]$ is defined in all of $\cH$, while the additional piece
\begin{equation}
\label{eqn:anom_anom_def}
\cA_{A}= \braket{\psi | (H^{\dagger} - H) A \psi} \,,
\end{equation}
is known as the anomaly. 

While we do not have a symmetry, we do have the commutator $\langle [H, \cO_{e}]\rangle = 0$ (where $\cO_{e} = S^\sigma Z^\zeta$), which may suffer anomalies. Since this is the only commutator needed to derive the recursion relations, we check if it receives anomalous corrections by calculating $\cA_{\cO_{e}}$.  Additionally, the states used in the anomaly calculations are energy eigenstates $\psi$, hence we use notation $\braket{\psi | A \psi } \equiv \langle A \rangle$.
Using integration by parts, the result for anomaly $\cA_{\cO_{e}}$ is
\begin{equation}
\label{eqn:anom_oe}
\cA_{\cO_{e}} = i^\sigma z(1-z) \left( \psi(z)^{*} \frac{\partial^{\sigma+1}}{\partial z^{\sigma +1}} (z^{\zeta} \psi(z)) - \psi'(z)^{*}\frac{\partial^{\sigma}}{\partial z^{\sigma }} (z^{\zeta} \psi(z)) \right) \bigg|_{0}^{1} \,.
\end{equation}
Unfortunately, we could not show that $\cA_{\cO_{e}}=0$ using the boundary conditions alone, but by employing the explicit energy eigenstate solutions, $\psi_{n}(z) = \psi_{n}(z)^{*} = \sqrt{2n+1} P_{n}(2z-1)$, it can be argued to vanish as follows. Since any derivative of $\psi_{n}(z)$ results in another polynomial, the contents of the large bracket in \eqref{eqn:anom_oe} will be a well-behaved polynomial in $z$, for any $n$. The factor of $z(1-z)$ outside the bracket then ensures that the total expression vanishes in the both limits $z \to 0$, $z \to 1$, therefore $\cA_{\cO_{e}} =0$.

As an aside, by following the alternative derivation of  \eqnverty in Appendix \ref{app:drr}, it is possible to show that this 1d relation is anomaly free using only the boundary conditions. The derivation utilises two commutator equations, namely
\begin{equation}
\label{eqn:anom_comm_eqns}
 \langle [H, \cO_{a}(Z)]\rangle =0 \,, \qquad \langle [H, \cO_{c}(S, Z)]\rangle =0
 \,,
\end{equation}
with operators $\cO_{a}= -\frac{i}{2} (\zeta-1) Z^{\zeta-1} (1-Z)$
 and $\cO_{c} = S Z^\zeta (1- Z)$. Therefore, by inserting these operators into \eqref{eqn:anom_anom_def}, we compute $\mathcal{A}_{\cO_{a}}$ and $\mathcal{A}_{\cO_{c}}$ using integration by parts
\begin{align}
\label{eqn:anom_oa}
\begin{split}
\cA_{\cO_{a}} = \frac{1}{2} i (\zeta -1)\bigg( z^{\zeta -1} (1-z) ( 1 - \zeta  (1-z)) |\psi (z)|^2 
+ z^{\zeta}(1-z)^2 \big(\psi'(z)^{*}\psi(z)- \psi(z)^{*}\psi'(z)  \big) \bigg)\bigg|_{0}^{1} \,,
\end{split}
\end{align}
\begin{align}
\begin{split}
\label{eqn:anom_oc}
\cA_{\cO_{c}} &= i\bigg( - z^{\zeta-1}(1-z)\left[Ez +\zeta(1+z -\zeta(1-z)) \right]|\psi(z)|^2 -z^{\zeta+1} (1-z)^{2} |\psi'(z)|^2 
\\
& \qquad  - z^{\zeta}(1-z)(1-2\zeta(1-z)) \psi(z)^{*} \psi'(z) + z^{\zeta}(1-z)(z-\zeta(1-z))\psi'(z)^{*} \psi(z) \bigg) \bigg|_{0}^{1}  \,,
\end{split}
\end{align}
Note, we have substituted the Schr{\"o}dinger equation into \eqref{eqn:anom_oc} to remove the second derivatives. Providing\footnote{Note that $\zeta \geq 1$ here simply ensures the operator $\cO_{a}$ contains no negative powers of $Z$. This is also relevant to the derived 1d recursion relation: under this choice, \eqnverty relates upper octant moments.} $\zeta \geq 1$ and that  evaluation at $z=0$, $z=1$ is exchanged for limits $z \to 0$, $z \to 1$ such that the boundary conditions of \eqref{eqn:model_bc_psi} can be applied, we find that $\cA_{\cO_{a}} = \cA_{\cO_{c}} = 0$.
Note here we have used that $\psi, \psi^{*}$ is finite according to equations \eqref{eqn:model_norm_cond_on_psi_z0} and \eqref{eqn:model_norm_cond_on_psi_z1}.
This reconfirms that the 1d recursion relation in $\langle Z^{\zeta} \rangle$-type moments does not receive anomalous corrections.

%%%%%%%%%%%%%%%%%%%%%%%%%%%%%%%%%%%%%%%%%%%%%%%%%%%%%%%%%%%%%%%%%%%%%%%%%%%%%%%%%%%%%%%%%SUBSECTION%%%%%%%%%%%%%%%%%%%%%%%%%%%%%%%%%%%%%%%%%%%%%%%%%%%%%%%%%%%%%%%%%%%%%%%%%%%%%%%%%%%%%%%%%
\subsection{Bootstrap matrices}
\label{ss:bootstrap_matrices}
This section details how the bootstrap matrices are constructed using the recursion relations, as well as the differences in the bootstrap operator ordering.

%%%%%%%%%%%%%%%%%%%%%%%%%%%%%%%%%%%%%%%%%%%%%%%%%%%%%%%%%%%%%%%%%%%%%%%%%%%%%%%%%%%%%%%%%SUBSECTION%%%%%%%%%%%%%%%%%%%%%%%%%%%%%%%%%%%%%%%%%%%%%%%%%%%%%%%%%%%%%%%%%%%%%%%%%%%%%%%%%%%%%%%%%

\subsubsection{One-operator matrix}
\label{sss:one_op_mat}
The 1d recursion relation in $\langle Z^{\zeta} \rangle$ moments, \eqnverty, provides the means to build a Hankel bootstrap matrix, $\mathcal{B}_{1d}$. We choose the bootstrap operator to be
\begin{equation}
\label{eqn:boot_b1d_o_def}
\cO= \sum_{\zeta \geq0} c_{\zeta} Z^{\zeta}
\end{equation}
with $c_{\zeta} \in \mathbb{C}$. Given $Z = Z^{\dagger}$, the elements of $\cB_{1d}$ defined from $\langle \cO^{\dagger} \cO \rangle = \sum_{\zeta, \zeta' \geq 0} c_{\zeta}(\cB_{1d})_{\zeta,\zeta'} c_{\zeta'} $ are
\begin{equation}
\label{eqn:boot_b1_mat_elemets}
(\mathcal{B}_{1d})_{\zeta \zeta'} = \left\langle \left(Z^{\zeta}\right)^{\dagger} Z^{\zeta'} \right\rangle = \langle Z^{\zeta+\zeta'} \rangle = f_{\zeta+\zeta'}  \,.
\end{equation}
The size of the matrix is hence determined by the order of the moments. For maximum index values $\zeta_{\text{max}}$ and $\zeta'_{\text{max}}$, we must obtain moments up to $f_{\zeta_{\text{max}}+\zeta'_{\text{max}}}$. The matrix size parameter is also defined in terms of these values, $K= \zeta_{\text{max}}+1 = \zeta'_{\text{max}}+1$. As an explicit example, choosing $(\zeta, \zeta') \in \{0,1,2\}$ (i.e. $K = 3$), we provide $\mathcal{B}_{1d}$ below populated with values obtained from the \eqnverty equation of \eqref{eqn:boot_1d_rec_in_z}.
\begin{equation}
\label{eqn:boot_b1_at_dim_4}
\mathcal{B}_{1d} = \begin{pmatrix}
f_{0} & f_{1} & f_{2} \\
f_{1} & f_{2} & f_{3} \\
f_{2} & f_{3} &f_{4} \\
\end{pmatrix}
= 
\begin{pmatrix}
1 & \frac{1}{2} & \frac{3 E-2 }{2(4E-3)} \\
\frac{1}{2} & \frac{3 E-2 }{2(4E-3)} & \frac{5 E-3}{4 (4 E-3)} \\
\frac{3 E-2 }{2(4E-3)} & \frac{5 E-3}{4 (4 E -3)} & \frac{5 E (7 E-30)+72 }{8 (4 E-15) (4 E-3)} \\
\end{pmatrix}
\end{equation}
where we have set $f_{0}=1$ as the chosen normalisation. Given that $Z$ is Hermitian, $\mathcal{B}_{1d}$ is real and symmetric and we also note that the matrix contains a single initial data, $E$. We explore the results of bootstrapping $\cB_{1d}$ at different matrix sizes in section \S \ref{s:numerical_results}.

%%%%%%%%%%%%%%%%%%%%%%%%%%%%%%%%%%%%%%%%%%%%%%%%%%%%%%%%%%%%%%%%%%%%%%%%%%%%%%%%%%%%%%%%%SUBSECTION%%%%%%%%%%%%%%%%%%%%%%%%%%%%%%%%%%%%%%%%%%%%%%%%%%%%%%%%%%%%%%%%%%%%%%%%%%%%%%%%%%%%%%%%%%%%%%%%%%%%%%%%%%%%%%%%%%%%%%%%%%%%%%%%%%%%%%%%%%%%%%%%%%%%%%%%%%

\subsubsection{Two-operator matrix}
\label{sss:two_op_matrix}

The 2d bootstrap matrix constructed here, $\cB_{2d}$, features expectation values of products of $S$ and $Z$ operators. These two-operator type matrices have previously been encountered in the literature, see \cite{Du:2021hfw,Hu:2022keu,Morita:2022zuy}. As such, we use the following bootstrap operator
\begin{equation}
\label{eqn:boot_main_op}
\cO = \sum_{\sigma, \zeta \geq 0} c_{\sigma,\zeta} S^{\sigma}[Z(1-Z)]^{\zeta} \,.
\end{equation}
where $c_{\sigma,\zeta} \in \mathbb{C}$. The choice of $\cO$ comes from an educated guess by first looking at the Hamiltonian which contains $Z(1-Z)$ and secondly looking at P{\"o}schl-Teller potentials as seen by changing coordinates\footnote{In \cite{Khan:2022uyz}, they analysed a system with P{\"o}schl-Teller potential using coordinates proportional to the $p$ coordinate in \S\ref{ss:po_tell_coords}, to develop a 1d recursion relation. Their relation focussed on moments $\left\langle \sech^{2 \zeta}(p/2) \right\rangle$ which, under the coordinate transformation of \eqref{eqn:model_z_to_p_coord}, are proportional to $\langle [Z(1-Z)]^{\zeta}\rangle$. Therefore, extending this idea by introducing the simplest dependence on $S$, we arrive at $\cO = \sum_{\sigma, \zeta \geq 0} c_{\sigma,\zeta} S^{\sigma}[Z(1-Z)]^{\zeta}$.}. Since $\langle \cO^{\dagger} \cO \rangle \geq 0$, we express the corresponding elements of $\cB_{2d}$ as
\begin{align}
\begin{split}
\label{eqn:boot_ele_odag_o}
(\cB_{2d})_{(\sigma, \zeta),(\sigma', \zeta')}  &= \langle \left(S^{\sigma} [Z(1-Z)]^{\zeta} \right)^{\dagger} S^{\sigma'} [Z(1-Z)]^{\zeta'}\rangle
\\
&=\langle[Z(1-Z)]^{\zeta} (S^{\sigma})^{\dagger}  S^{\sigma'} [Z(1-Z)]^{\zeta'}\rangle
\end{split}
\end{align}
such that $\cB_{2d} \succeq 0$. Importantly, to evaluate this expression we must take care since $S$ is not self-adjoint and the recursion relations only contain $S$, not $S^{\dagger}$. The following lemma bypasses this issue.
\begin{namedtheorem}[``Dagger"]
\label{lem:boot_lemma_main_text}
For $\beta \geq \alpha \geq 0$, the following is true
\begin{equation}
\label{eqn:boot_lemma_rem_dag_eqn}
[Z(1-Z)]^{\beta}(S^{\alpha})^{\dagger} = [Z(1-Z)]^{\beta}S^{\alpha} \,.
\end{equation}
\end{namedtheorem}

The proof of this lemma is provided in Appendix \ref{app:lemma_proof_dag_remove}. It should also be noted that initial numerical checks of operator $\cO$ using the analytic solutions, encouraged investigation into the proof of \eqref{eqn:boot_lemma_rem_dag_eqn}. Upon applying the lemma, using the binomial expansion for $(1-Z)^{\zeta}$ and $(1-Z)^{\zeta'}$ then applying the McCoy formula \cite{mccoyref},  the element in \eqref{eqn:boot_ele_odag_o} becomes
\begin{align}
\begin{split}
\label{eqn:boot_odago_ele_expanded_mccoy}
(\cB_{2d})_{(\sigma,\zeta),(\sigma',\zeta')}
&= \sum_{\kappa=0}^{\zeta} \sum_{\kappa'=0}^{\zeta'} (-1)^{\kappa+\kappa'} \binom{\zeta}{\kappa}\binom{\zeta'}{\kappa'} \Bigg( \langle S^{\sigma + \sigma'} Z^{\zeta+\zeta'+\kappa+\kappa'} \rangle 
\\
& + \sum_{\lambda=1}^{\text{min}(\zeta+\kappa,\sigma+\sigma')} \frac{(-i)^{\lambda}(\sigma + \sigma')!(\zeta+\kappa)! }{\lambda! (\sigma+\sigma' - \lambda)! (\zeta+\kappa -\lambda)!} \langle S^{\sigma+\sigma' - \lambda}Z^{\zeta+\zeta'+\kappa+\kappa' -\lambda} \rangle\Bigg)
\end{split}
\end{align}
The rows of matrix $\mathcal{B}_{2d}$ are indexed by the tuple $(\sigma,\zeta)$ and the columns by $(\sigma',\zeta')$. Another reason why $\cO = S^{\sigma}[Z(1-Z)]^{\zeta}$ is employed is seen from this binomial expansion form of the matrix element: it is a sum of $f_{\sigma,\zeta}$-type moments, which can be readily found from the recursion relations of \S \ref{ss:rec_rels}.

$\cB_{2d}$ can be built from any general combination of $\sigma, \zeta, \sigma', \zeta'$, providing the Dagger lemma is satisfied and the appropriate restrictions in the $(\sigma, \zeta)$ space are used, according to \S \ref{ss:octant_choice}. We choose to focus on a simple subset of $(\sigma, \zeta)$, constructing two matrices with the following element definitions
\begin{align}
\label{eqn:boot_b2d_mat_restrictions}
\begin{split}
(\cB_{2d}')_{(\sigma, \zeta),(\sigma', \zeta')} = (\cB_{2d})_{(\sigma, \zeta),(\sigma', \zeta')} & \quad \text{ for all } \quad \zeta = \sigma \,, \quad \zeta' = \sigma' \,, \quad \sigma, \sigma' \geq 0 \,,
\\
(\cB_{2d}'')_{(\sigma, \zeta),(\sigma', \zeta')} = (\cB_{2d})_{(\sigma, \zeta),(\sigma', \zeta')} & \quad \text{ for all } \quad \zeta = 2\sigma \,,  \quad \zeta' = 2\sigma' \,, \quad \sigma, \sigma' \geq 0 \,.
\end{split}
\end{align}
Explicit examples of $\cB_{2d}'$ and $\cB_{2d}''$ at size $K=3$ are given below
\begin{flalign}
\label{eqn:boot_b2dp}
\cB_{2d}' &=
\begin{pmatrix}
(\cB_{2d}')_{(0,0),(0,0)} & (\cB_{2d}')_{(0,0),(1,1)} & (\cB_{2d}')_{(0,0),(2,2)} \\
(\cB_{2d}')_{(1,1),(0,0)} & (\cB_{2d}')_{(1,1),(1,1)} & (\cB_{2d}')_{(1,1),(2,2)} \\
(\cB_{2d}')_{(2,2),(0,0)} & (\cB_{2d}')_{(2,2),(1,1)} & (\cB_{2d}')_{(2,2),(2,2)} \\
\end{pmatrix} && \nonumber
\\
&=
\begin{pmatrix}
 1 & 0 & \frac{E(E-2) }{2(4E-3)} \\
 0 & \frac{E(E-2) }{2(4E-3)} & 0 \\
 \frac{E(E-2) }{2(4E-3)} & 0 & \frac{3 (E(E-4) (E+2) (E+14)+96)}{8 (4E-15) (4E-3)} \\
\end{pmatrix} &&
\end{flalign}
\begin{align}
\label{eqn:boot_b2dpp}
&\hspace*{-1cm}\mathcal{B}_{2d}'' =
\begin{pmatrix}
(\cB_{2d}'')_{(0,0),(0,0)} & (\cB_{2d}'')_{(0,0),(1,2)} & (\cB_{2d}'')_{(0,0),(2,4)} \\
(\cB_{2d}'')_{(1,2),(0,0)} & (\cB_{2d}'')_{(1,2),(1,2)} & (\cB_{2d}'')_{(1,2),(2,4)} \\
(\cB_{2d}'')_{(2,4),(0,0)} & (\cB_{2d}'')_{(2,4),(1,2)} & (\cB_{2d}'')_{(2,4),(2,4)} \\
\end{pmatrix} \nonumber
\\
&\hspace*{-1cm}= {\tiny
\left(\hspace*{-0.2cm}
\begin{array}{ccc}
 1 & \hspace{-1cm}  0 & \hspace{-1cm}  \frac{E \left(5 E^3-68 E^2+236 E-240\right)}{16 (4 E-35) (4 E-15) (4 E-3)} \\
 0 & \hspace{-1cm}  \frac{5 E^4-52 E^3+20 E^2+512 E-480}{16 (4 E-35) (4 E-15) (4 E-3)} & \hspace{-1cm} 0 \\
 \frac{E \left(5 E^3-68 E^2+236 E-240\right)}{16 (4 E-35) (4 E-15) (4 E-3)} & \hspace{-1cm}  0 & \hspace{-1cm} \frac{3 \left(77 E^8-4872 E^7+46536 E^6+2214656 E^5-49787184 E^4+286032000 E^3+25484544 E^2-2505572352 E+2299207680\right)}{1024 (4 E-143) (4 E-99) (4 E-63) (4 E-35) (4 E-15) (4 E-3)} \\
\end{array}
\hspace*{-0.2cm}\right)
}
\end{align}
where we have set $f_{0,0} = 1$. Note that these are not Hankel matrices, as $\cO$ contains both $S$ and $Z$. The matrix elements in both cases were obtained utilising three of the recursion relations: \eqnverty, \eqndiagy and \eqnupy . Surprisingly, we see that although this set of recursion relations use initial data $\{E,  f_{1,1}\}$, the matrices themselves only depend on the energy $E$. This is due to cancellations of the $f_{1,1}$ moments in the matrix elements. This makes the matrices real and symmetric: a fact that does not necessarily hold true when initial data $f_{1,1}$ is also present.

%%%%%%%%%%%%%%%%%%%%%%%%%%%%%%%%%%%%%%%%%%%%%%%%%%%%%%%%%%%%%%%%%%%%%%%%%%%%%%%%%%%%%%%%%SUBSECTION%%%%%%%%%%%%%%%%%%%%%%%%%%%%%%%%%%%%%%%%%%%%%%%%%%%%%%%%%%%%%%%%%%%%%%%%%%%%%%%%%%%%%%%%%%%%%%%%%%%%%%%%%%%%%%%%%%%%%%%%%%%%%%%%%%%%%%%%%%%%%%%%%%%%

\subsubsection{Alternative two-operator matrix}
\label{sss:alt_two_op_matrix}
We can generate an alternative bootstrap matrix, $\tilde{\cB}_{2d}$, by switching the operator order from  $\langle \cO^{\dagger} \cO \rangle$ to $\langle \cO \cO^{\dagger} \rangle$.
In doing so, we restrict the range of the bootstrap operator $\cO$ further\footnote{
As a note of clarification, the operator indices $\sigma$ and $\zeta$ appear in two different contexts. The first instance is on operator $S^{\sigma} Z^{\zeta}$ which appear via $f_{\sigma,\zeta}$ in the recursion relations. The second instance is on the bootstrap matrix operator $S^{\sigma} [Z(1-Z)]^{\zeta}$. In the second instance, since the bootstrap matrix elements can be decomposed via the McCoy formula into sums over $f_{\sigma, \zeta}$-type moments, see \eqref{eqn:boot_odago_ele_expanded_mccoy} and \eqref{eqn:boot_oodag_ele_expanded_mccoy}, then these elements also follow the same restrictions that apply to the first instance. The additional restrictions of \eqref{eqn:boot_b2d_mat_restrictions}, which are used in both $\langle \cO^{\dagger} \cO \rangle$ and $\langle \cO \cO^{\dagger} \rangle$ cases, hence apply to the associated $f_{\sigma, \zeta}$ that appear and this is demonstrated by the darker green region ($\zeta \geq \sigma \geq 0$) in Figure \ref{fig:octant_with_eqns}.}
\begin{equation}
\label{eqn:boot_main_op_furt_res}
\cO = \sum_{\sigma, \zeta \in \Omega} c_{\sigma,\zeta} S^{\sigma}[Z(1-Z)]^{\zeta} \,,
\end{equation}
where $c_{\sigma,\zeta} \in \mathbb{C}$. Here, $\Omega$ refers to the positive upper octant: $\Omega = \{ \zeta \geq \sigma \geq 0 \}$. This restriction is applied since the expectation values appearing in the $\tilde{\cB}_{2d}$ matrix elements are shown to be finite in this region. By contrast, in the positive lower octant, $(\sigma > \zeta \geq 0)$, they can blow up. The proof of this statement may be found in Appendix \ref{app:fme_main_proof}, where we insert a complete set of energy eigenstates $\psi_{n}$ to show that the finiteness is linked to the Wigner (3j) coefficient and its selection rules. We provide explicit examples of positive, lower and upper octant calculations in Appendix \ref{app:fme_up_low_examples}, to show how these infinities arise. Also, depending on the insertion point of the complete set of states (i.e. between which operators they are inserted), infinities can arise in the both upper and lower positive octant cases. However, by applying regularisation, only the upper positive octant appears to yield consistent outcomes, as can be seen in Appendix \ref{app:fme_regularise}.

From the $\cO$ defined in \eqref{eqn:boot_main_op_furt_res}, the alternative 2d matrix, $\tilde{\cB}_{2d}$, therefore has matrix elements
\begin{align}
\begin{split}
\label{eqn:boot_ele_oodag}
\left(\tilde{\cB}_{2d}\right)_{(\sigma,\zeta),(\sigma',\zeta')}
&= \left\langle S^{\sigma}[Z(1-Z)]^{\zeta} \left( S^{\sigma'}[Z(1-Z)]^{\zeta'} \right)^{\dagger} \right\rangle
\\
&= \left\langle S^{\sigma}[Z(1-Z)]^{\zeta} [Z(1-Z)]^{\zeta'} \left( S^{\sigma'} \right)^{\dagger} \right\rangle
\\
&= \left\langle S^{\sigma}[Z(1-Z)]^{\zeta + \zeta'} S^{\sigma'} \right\rangle \,.
\end{split}
\end{align}
where the Dagger lemma  \eqref{eqn:boot_lemma_rem_dag_eqn} was once again employed to remove the $\dagger$ from $(S^{\sigma'})^{\dagger}$, and thus, $\zeta+\zeta' \geq \sigma'$ must be enforced throughout the calculation. Using the commutation relation, binomial expansion and the McCoy formula we bring it to its final form, where like $\cB_{2d}$, each element is a sum of $f_{\sigma,\zeta}$-type expectation values,
\begin{align}
\label{eqn:boot_oodag_ele_expanded_mccoy}
\begin{split}
&\left(\tilde{\cB}_{2d}\right)_{(\sigma,\zeta),(\sigma',\zeta')}
\\
&=\sum_{\kappa=0}^{\zeta+\zeta'}(-1)^{\kappa} \binom{\zeta + \zeta'}{\kappa} \left\langle S^{\sigma + \sigma'} Z^{\zeta+\zeta' +\kappa} +  \hspace{-0.3cm} \sum_{\lambda=1}^{\text{min}(\sigma',\,\zeta+\zeta'+\kappa)} \hspace{-0.3cm} \frac{(-i)^{\lambda}(\sigma')!(\zeta+\zeta' +\kappa)!  }{\lambda! (\sigma'-\lambda)! (\zeta+\zeta'+\kappa-\lambda)!} S^{\sigma + \sigma'-\lambda} Z^{\zeta+\zeta'+\kappa-\lambda}  \right\rangle \,.
\end{split}
\end{align}
We establish two unique matrices, $\tilde{\cB}_{2d}'$ and $\tilde{\cB}_{2d}''$ which are defined by applying the same subsets of indices as presented in \eqref{eqn:boot_b2d_mat_restrictions} for $\cB_{2d}$. These $(\sigma, \zeta)$ choices are in $\Omega$ and also ensure the lemma index constraint is automatically satisfied.
Taking matrix size parameter $K=3$, the explicit matrices are given below
\begin{align}
\label{eqn:boot_tilb2dp}
\tilde{\cB}_{2d}' &=
\begin{pmatrix}
(\tilde{\cB}_{2d}')_{(0,0),(0,0)} & (\tilde{\cB}_{2d}')_{(0,0),(1,1)} & (\tilde{\cB}_{2d}')_{(0,0),(2,2)} \\
(\tilde{\cB}_{2d}')_{(1,1),(0,0)} & (\tilde{\cB}_{2d}')_{(1,1),(1,1)} & (\tilde{\cB}_{2d}')_{(1,1),(2,2)} \\
(\tilde{\cB}_{2d}')_{(2,2),(0,0)} & (\tilde{\cB}_{2d}')_{(2,2),(1,1)} & (\tilde{\cB}_{2d}')_{(2,2),(2,2)} \\
\end{pmatrix} = \begin{pmatrix}
1 & 0 & \frac{E(E-2)}{2(4E-3)}
\\
0 & \frac{E^2}{2(4E-3)} & 0
 \\
\frac{E(E-2)}{2(4E-3)} & 0 & \frac{3 E^2 (E-2)^2}{8(4E-15)(4E-3)} \\
\end{pmatrix}
\\
\label{eqn:boot_tilb2dpp}
\tilde{\cB}_{2d}'' &=
\begin{pmatrix}
(\tilde{\cB}_{2d}'')_{(0,0),(0,0)} & (\tilde{\cB}_{2d}'')_{(0,0),(1,2)} & (\tilde{\cB}_{2d}'')_{(0,0),(2,4)} \\
(\tilde{\cB}_{2d}'')_{(1,2),(0,0)} & (\tilde{\cB}_{2d}'')_{(1,2),(1,2)} & (\tilde{\cB}_{2d}'')_{(1,2),(2,4)} \\
(\tilde{\cB}_{2d}'')_{(2,4),(0,0)} & (\tilde{\cB}_{2d}'')_{(2,4),(1,2)} & (\tilde{\cB}_{2d}'')_{(2,4),(2,4)} \\
\end{pmatrix} \nonumber
\\
&= {\tiny
 \left(
\begin{array}{ccc}
 1 & 0 & \frac{E \left(5 E^3-68 E^2+236 E-240\right)}{16 (4 E-35) (4 E-15) (4 E-3)} \\
 0 & \frac{E^2 \left(5 E^2-64 E+180\right)}{16 (4 E-35) (4 E-15) (4 E-3)} & 0 \\
 \frac{E \left(5 E^3-68 E^2+236 E-240\right)}{16 (4 E-35) (4 E-15) (4 E-3)} & 0 & \frac{21 (E-2)^2 E^2 \left(11 E^4-940 E^3+27948 E^2-340224 E+1425600\right)}{1024 (4 E-143) (4 E-99) (4 E-63) (4 E-35) (4 E-15) (4 E-3)} \\
\end{array}
\right)
}
\end{align}
where we have set $f_{0,0} = 1$. The same three recursion relations \eqnverty, \eqndiagy and \eqnupy are used to generate all element moments, and the initial data is once more just the energy $E$, due to $f_{1,1}$ cancellations.

%%%%%%%%%%%%%%%%%%%%%%%%%%%%%%%%%%%%%%%%%%%%%%%%%%%%%%%%%%%%%%%%%%%%%%%%%%%%%%%%%%%%%%%%%SUBSECTION%%%%%%%%%%%%%%%%%%%%%%%%%%%%%%%%%%%%%%%%%%%%%%%%%%%%%%%%%%%%%%%%%%%%%%%%%%%%%%%%%%%%%%%%%

%%%%%%%%%%%%%%%%%%%%%%%%%%%%%%%%%%%%%%%%%%%%%%%%%%%%%%%%%%%%%%%%%%%%%%%%%%%%%%%%%%%%%%%%%%%%%%%%%%%%%%%%%%SECTION%%%%%%%%%%%%%%%%%%%%%%%%%%%%%%%%%%%%%%%%%%%%%%%%%%%%%%%%%%%%%%%%%%%%%%%%%%%%%%%%%%%%%%%%%%%%%%%%%%%%%%%%%%%%%%%%%

\section{Results}
\label{s:numerical_results}

With the bootstrap matrices constructed, we now test that they are positive semi-definite on a selection of initial data. For all five types of matrices considered, $\{ \cB_{1d}, \cB_{2d}',\cB_{2d}'', \tilde{\cB}_{2d}', \tilde{\cB}_{2d}'' \}$, the initial data\footnote{To reiterate, the matrices generally depend on $\{E,f_{1,1}\}$, but due to cancellations of $f_{1,1}$ in $\cB_{2d}$ and $\tilde{\cB}_{2d}$, the energy is the only initial data in the discussed cases.} is simply the energy $E$. Given this one-dimensional search space, we scan over energy values in the interval $E \in [0,50]$, with a step size of $\Delta E = 10^{-2}$ for matrix sizes $ 2 \leq K \leq 8$ (for $K \times K$ bootstrap matrices). In Appendix \ref{app:slack_var_method} we present a self-contained description of an alternative, semi-definite programming search method that uses a so-called slack variable. 

%%%%%%%%%%%%%%%%%%%%%%%%%%%%%%%%%%%%%%%%%%%%%%%%%%%%%%%%%%%%%%%%%%%%%%%%%%%%%%%%%%%%%%%%%%%%%%%%%%%%%%%%%%SUBSECTION%%%%%%%%%%%%%%%%%%%%%%%%%%%%%%%%%%%%%%%%%%%%%%%%%%%%%%%%%%%%%%%%%%%%%%%%%%%%%%%%%%%%%%%%%%%%%%%%%%%%%%%%%%%%%%%%%

\subsection{One-operator matrix}
\label{ss:b1d_results}
For $\cB_{1d}$, the results are presented in Figure \ref{fig:1d_vertical_energy_vs_matdim}, plotting matrix size parameter $K$ vs. energy $E$. Here the black dashed lines indicate the energy eigenvalues of the model, which are $E_{n}= n(n+1)$ for $n=0,1,2 \dots$ and calculated using the analytic solution integration of $\langle \psi_{n} | H | \psi_{n} \rangle$. The red crosses denote singularities\footnote{Note, the proof in Appendix \ref{app:fme_main_proof} of finite matrix elements applies to eigenstates and therefore finiteness is only guaranteed at the energy eigenvalues.} in the matrix at a given $K$. The primary result for this matrix is that the allowed energies $E$ are confined to bands. The figure shows that bands for larger energy eigenvalues emerge as we increase $K$. Although some bands shrink as $K$ is increased, it is possible that they remain finite at $K=\infty$. For example the bands around $E=2$ decrease between $K=3$ and $K=8$, but slowly relative to the size of the matrix.

\begin{figure}[t!]
\centering
\includegraphics[width=14cm, height=6.8cm]{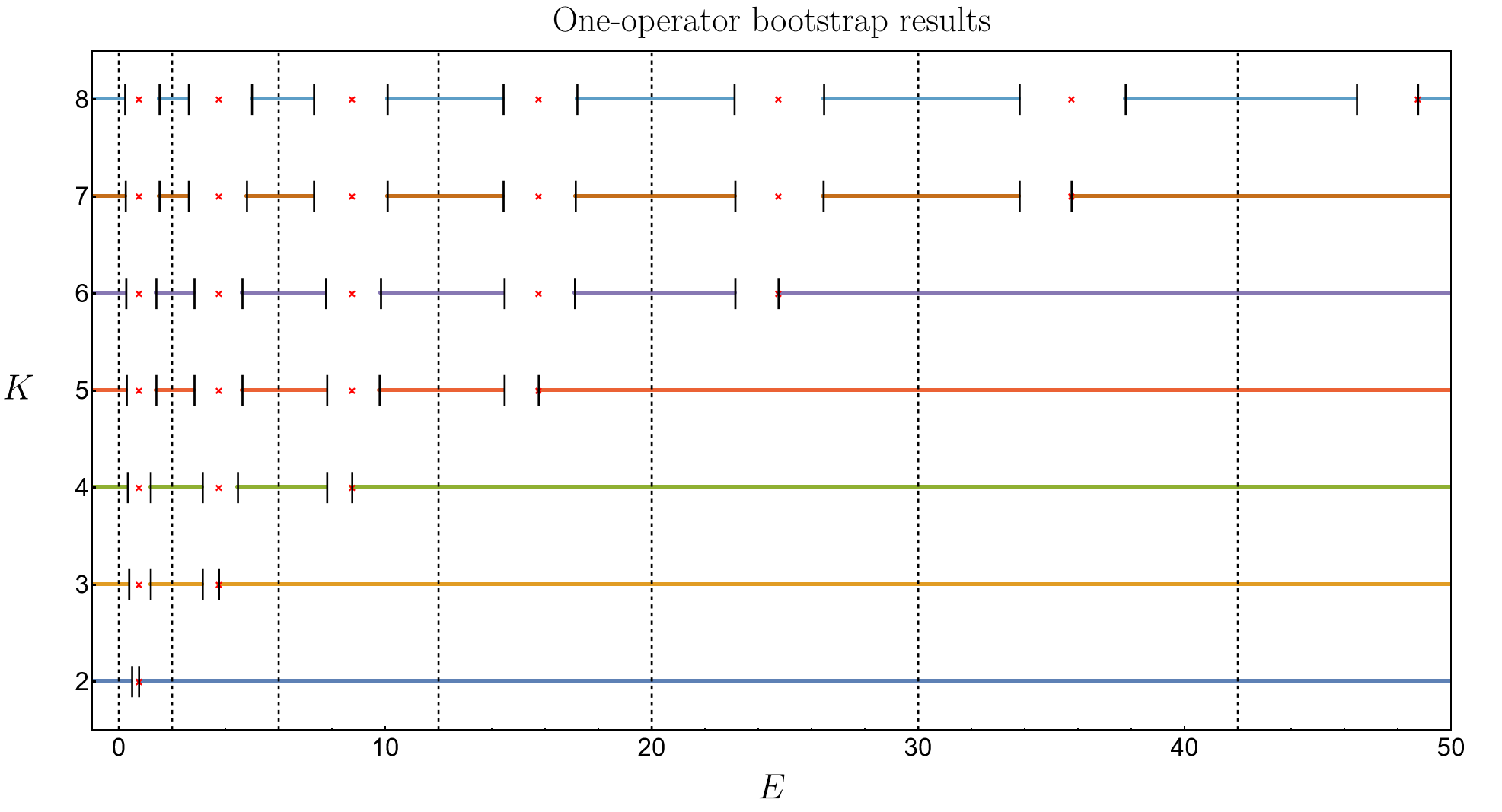}
\captionsetup{skip=10pt}
\caption{Matrix size parameter $K$ vs. energy $E$ for the bootstrap matrix $\cB_{1d}$. Black dashed lines indicate energy eigenvalues, while red crosses indicate those $E$ values which produce singularities in the matrix elements. At larger $K$, we see more bands forming, revealing localised regions around the energy eigenvalues. The plot also indicates the results for small negative $E$.}
\label{fig:1d_vertical_energy_vs_matdim}
\end{figure}

%%%%%%%%%%%%%%%%%%%%%%%%%%%%%%%%%%%%%%%%%%%%%%%%%%%%%%%%%%%%%%%%%%%%%%%%%%%%%%%%%%%%%%%%%%%%%%%%%%%%%%%%%%SUBSECTION%%%%%%%%%%%%%%%%%%%%%%%%%%%%%%%%%%%%%%%%%%%%%%%%%%%%%%%%%%%%%%%%%%%%%%%%%%%%%%%%%%%%%%%%%%%%%%%%%%%%%%%%%%%%%%%%%

\subsection{Two-operator matrix}
\label{ss:b2d_results}

Here the results of the bootstrap for matrices $\cB_{2d}'$ and $\cB_{2d}''$ are presented in Figures \ref{fig:b2dp_step} and \ref{fig:b2dpp_step} respectively. The black dashed lines refer to energy eigenvalues and red crosses to matrix singularities. In both cases, the bootstrap performs poorly, unable to constrain the energy eigenvalues effectively. The $\cB_{2d}''$ performs slightly better, with smaller bands appearing compared to $\cB_{2d}'$, for example at $K=8$ around $E=2$. For larger energies, these $\cB_{2d}$ cases provide almost no restriction on the  energies. This is seen by how close the vertical black bars (ends of the bands) are to the singularities and for $E>20$, the bootstrap essentially rules no energies out.

\begin{figure}[t!]
\begin{center}
\includegraphics[width=14cm, height=6.8cm]{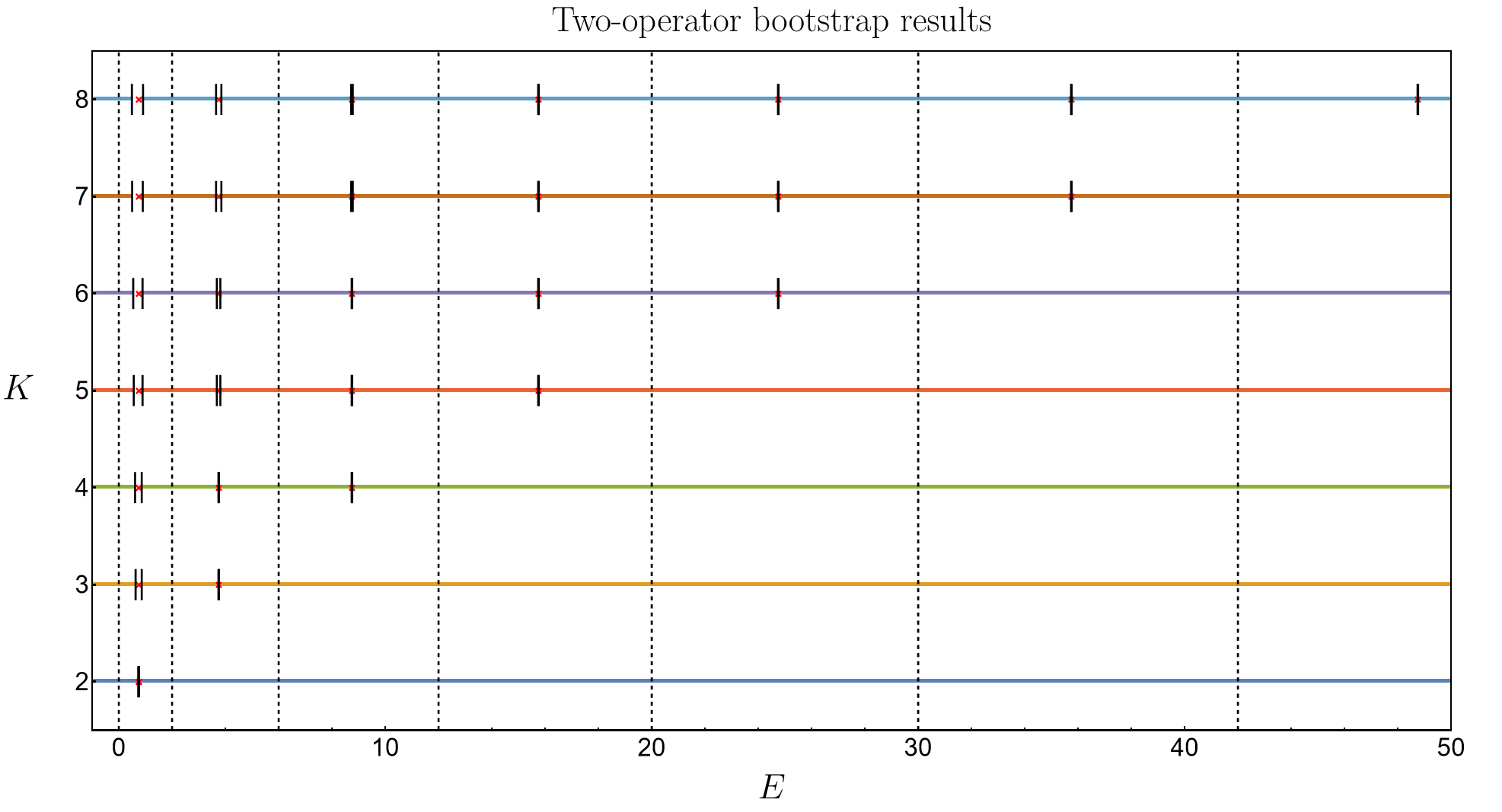}
\caption{$K$ vs. $E$ for matrix $\cB_{2d}'$. The plot shows that the bootstrap fails to constrain the energy eigenvalues of the system. Even for $K=8$, the bands remain large, with a performance worse than the $\cB_{1d}$ case. Red crosses represent matrix singularities, and black dashed lines are energy eigenvalues. }
\label{fig:b2dp_step}
\end{center}
\end{figure}

\begin{figure}[t!]
\begin{center}
\includegraphics[width=14cm, height=6.8cm]{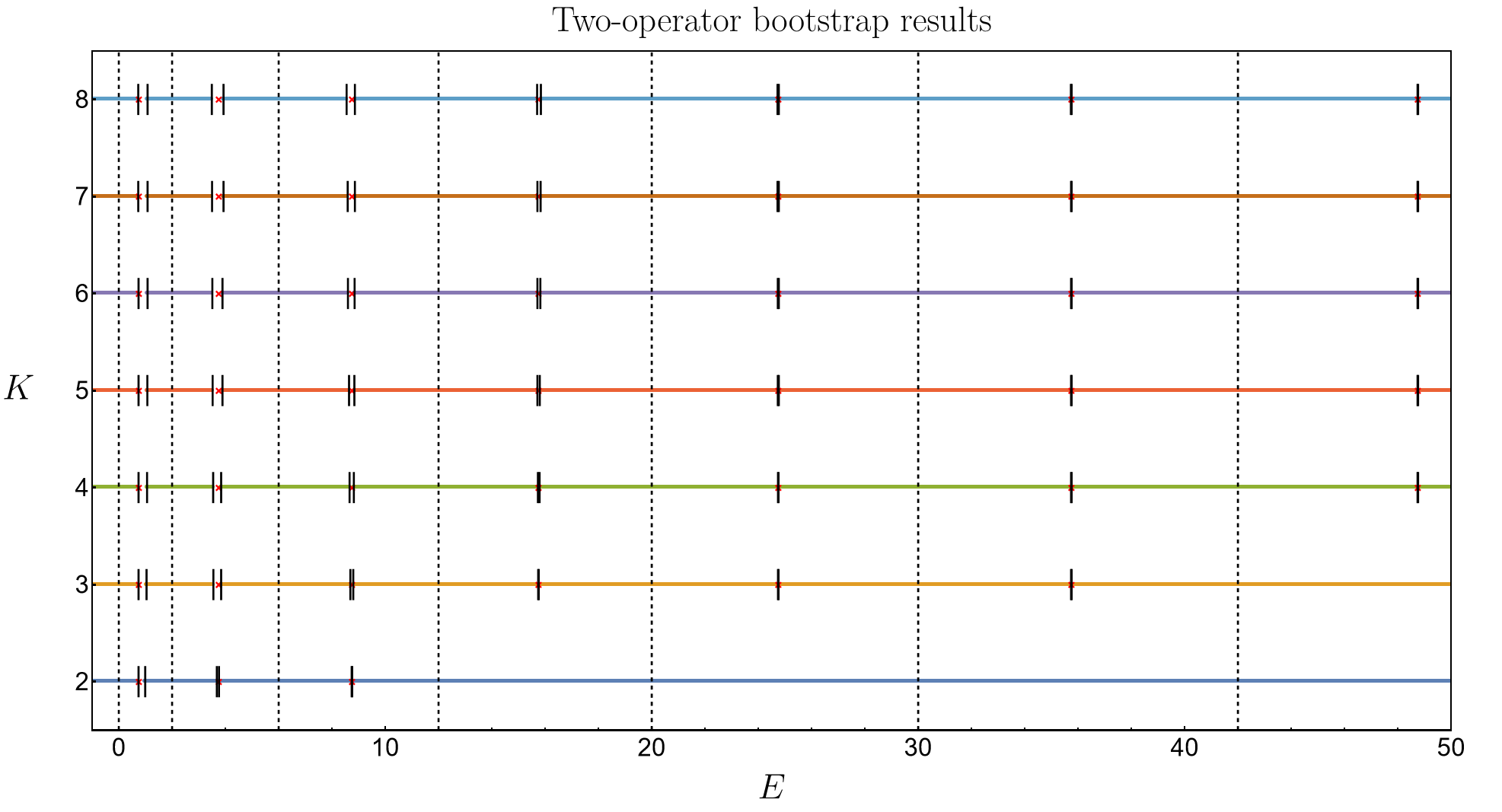}
\caption{$K$ vs. $E$ for matrix $\cB_{2d}''$. The bootstrap performs poorly in isolating the energy eigenvalues. We see slight improvement compared to the $\cB_{2d}'$ case in Figure \ref{fig:b2dp_step}, but the band size reduction with increasing $K$ is still minimal.}
\label{fig:b2dpp_step}
\end{center}
\end{figure}

%%%%%%%%%%%%%%%%%%%%%%%%%%%%%%%%%%%%%%%%%%%%%%%%%%%%%%%%%%%%%%%%%%%%%%%%%%%%%%%%%%%%%%%%%%%%%%%%%%%%%%%%%%SUBSECTION%%%%%%%%%%%%%%%%%%%%%%%%%%%%%%%%%%%%%%%%%%%%%%%%%%%%%%%%%%%%%%%%%%%%%%%%%%%%%%%%%%%%%%%%%%%%%%%%%%%%%%%%%%%%%%%%%

\subsection{Alternative two-operator matrix}
\label{ss:tildeb2d_results}

Figures \ref{fig:tilb2dp_step} and \ref{fig:tilb2dpp_step} display the most significant finding of the paper, corresponding to the bootstrap results of matrices $\tilde{\cB}_{2d}'$ and $\tilde{\cB}_{2d}''$. In both cases, the bootstrap is capable of identifying the lower-lying energy eigenvalues of the Hamiltonian exactly, up to the chosen step size. Figure \ref{fig:tilb2dp_step} shows that at a given $K$, a certain number of energy eigenvalues are precisely located and are indicated by black crosses. Then, at energies larger than a particular singularity, a single energy band remains. 
%For example, the $K=5$ result resolves the $E=0,2,6,12$ energy eigenvalues exactly, and immediately after the singularity at $E=15.75$, no more energies can be ruled out. 
Figure \ref{fig:tilb2dpp_step} associated to $\tilde{\cB}_{2d}''$, locates the same energy eigenvalues, but also constrains a selection of the larger energies into bands. The number of energy levels found exactly in both cases, appears to be equal to $K-1$. This alternative ordering therefore greatly outperforms the ordering encountered in both $\cB_{1d}$ and $\cB_{2d}$.

\begin{figure}[t!]
\begin{center}
\includegraphics[width=14cm, height=6.8cm]{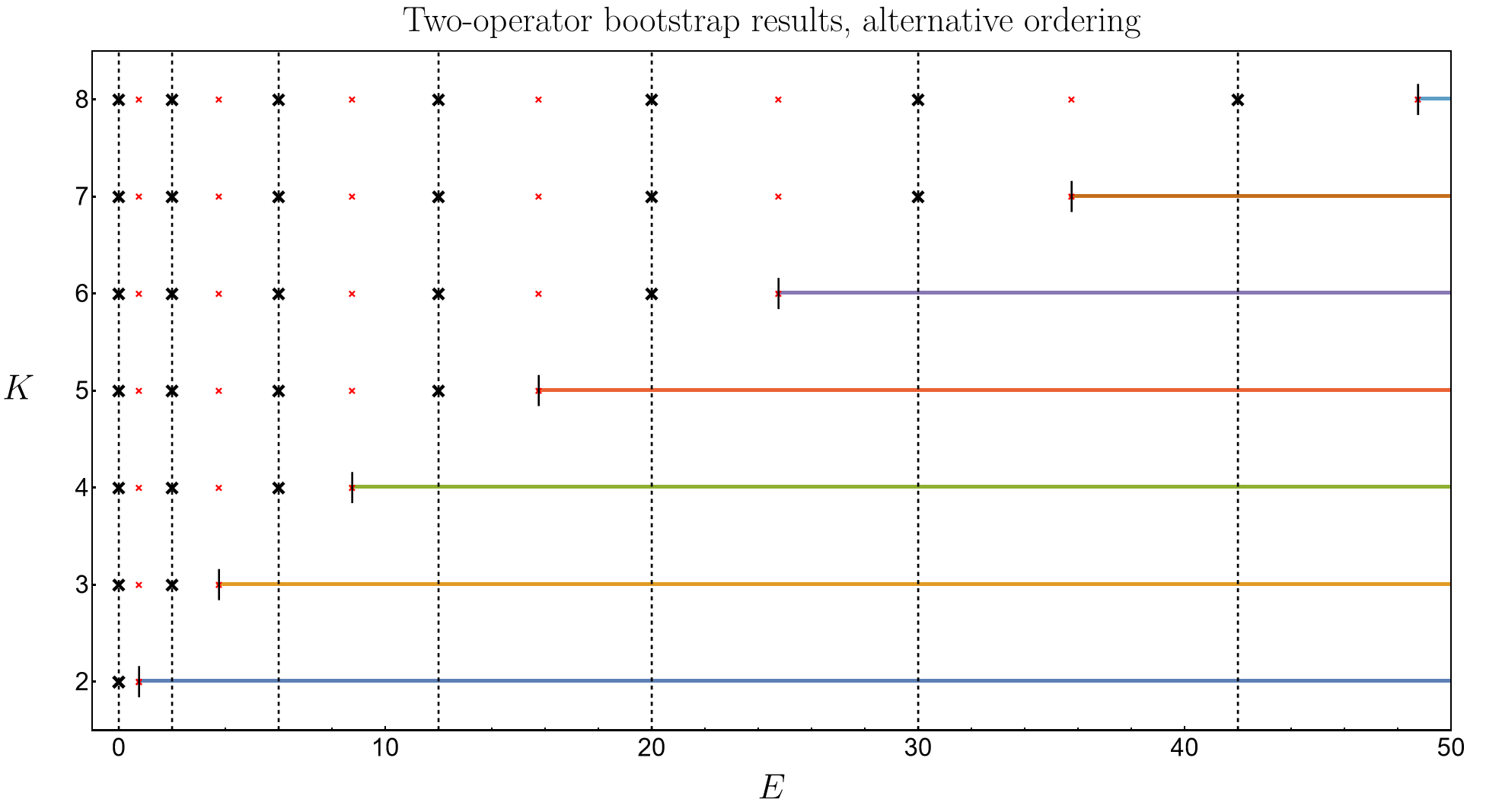}
\caption{$K$ vs. $E$ for matrix $\tilde{\cB}_{2d}'$. The coloured bands/black crosses correspond to the regions/points in the energy space which make $\tilde{\cB}'_{2d} \succeq 0$. The black crosses are the most significant result, indicating that the bootstrap is capable of identifying a number of energy eigenvalues exactly. }
\label{fig:tilb2dp_step}
\end{center}
\end{figure}

\begin{figure}[t!]
\begin{center}
\includegraphics[width=14cm,height=6.8cm]{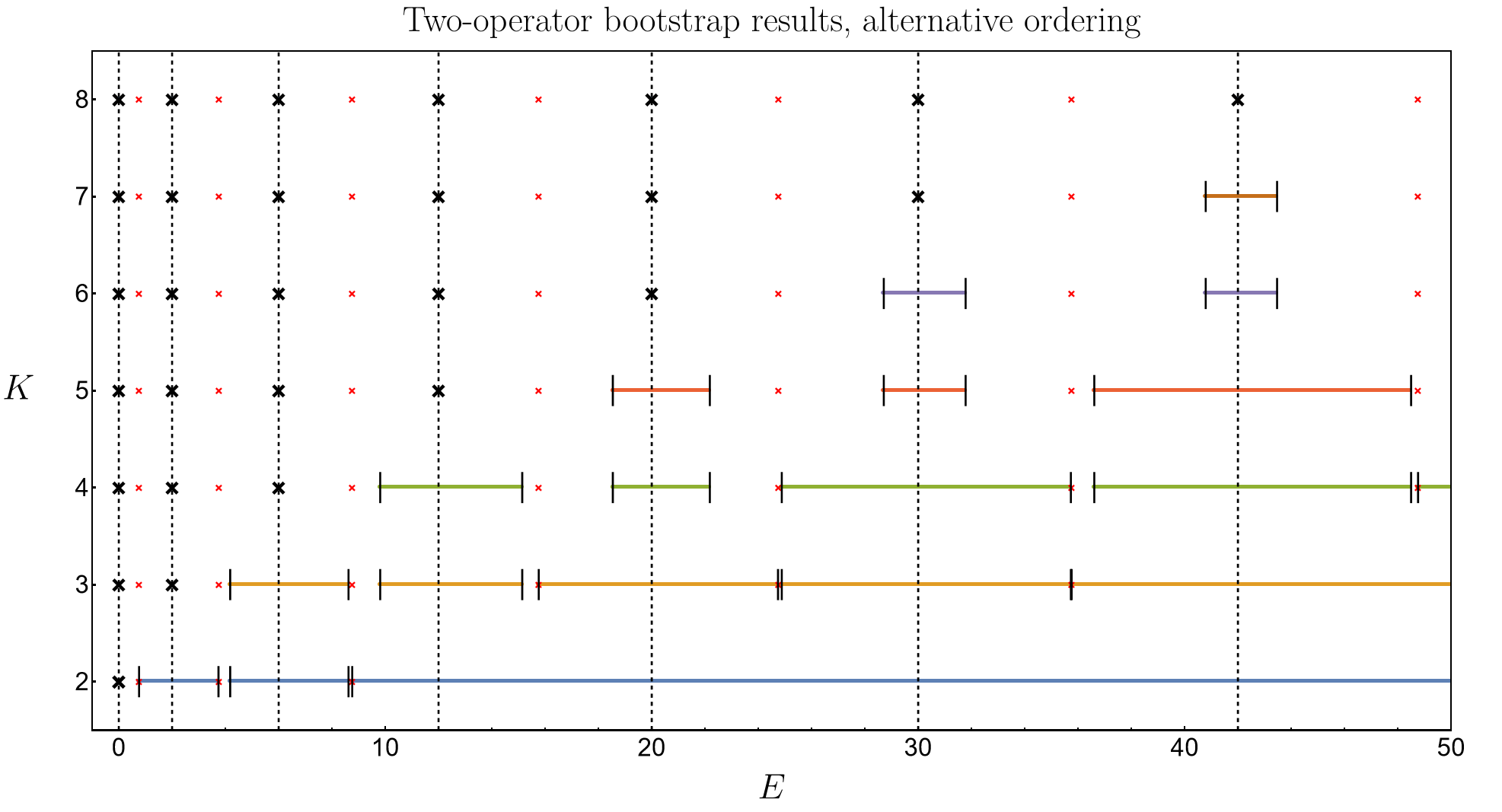}
\caption{$K$ vs. $E$ for matrix $\tilde{\cB}_{2d}''$. Here we find that not only are a set of energy eigenvalues identified exactly for a given $K$, but larger values of $E$ are also constrained into bands. For every integer increase in $K$, we appear to obtain an additional energy eigenvalue.}
\label{fig:tilb2dpp_step}
\end{center}
\end{figure}

%%%%%%%%%%%%%%%%%%%%%%%%%%%%%%%%%%%%%%%%%%%%%%%%%%%%%%%%%%%%%%%%%%%%%%%%%%%%%%%%%%%%%%%%%%%%%%%%%%%%%%%%%%SUBSECTION%%%%%%%%%%%%%%%%%%%%%%%%%%%%%%%%%%%%%%%%%%%%%%%%%%%%%%%%%%%%%%%%%%%%%%%%%%%%%%%%%%%%%%%%%%%%%%%%%%%%%%%%%%%%%%%%%
\clearpage

\begin{figure}[htb!]
\begin{center}
\includegraphics[width=10cm, height=6cm]{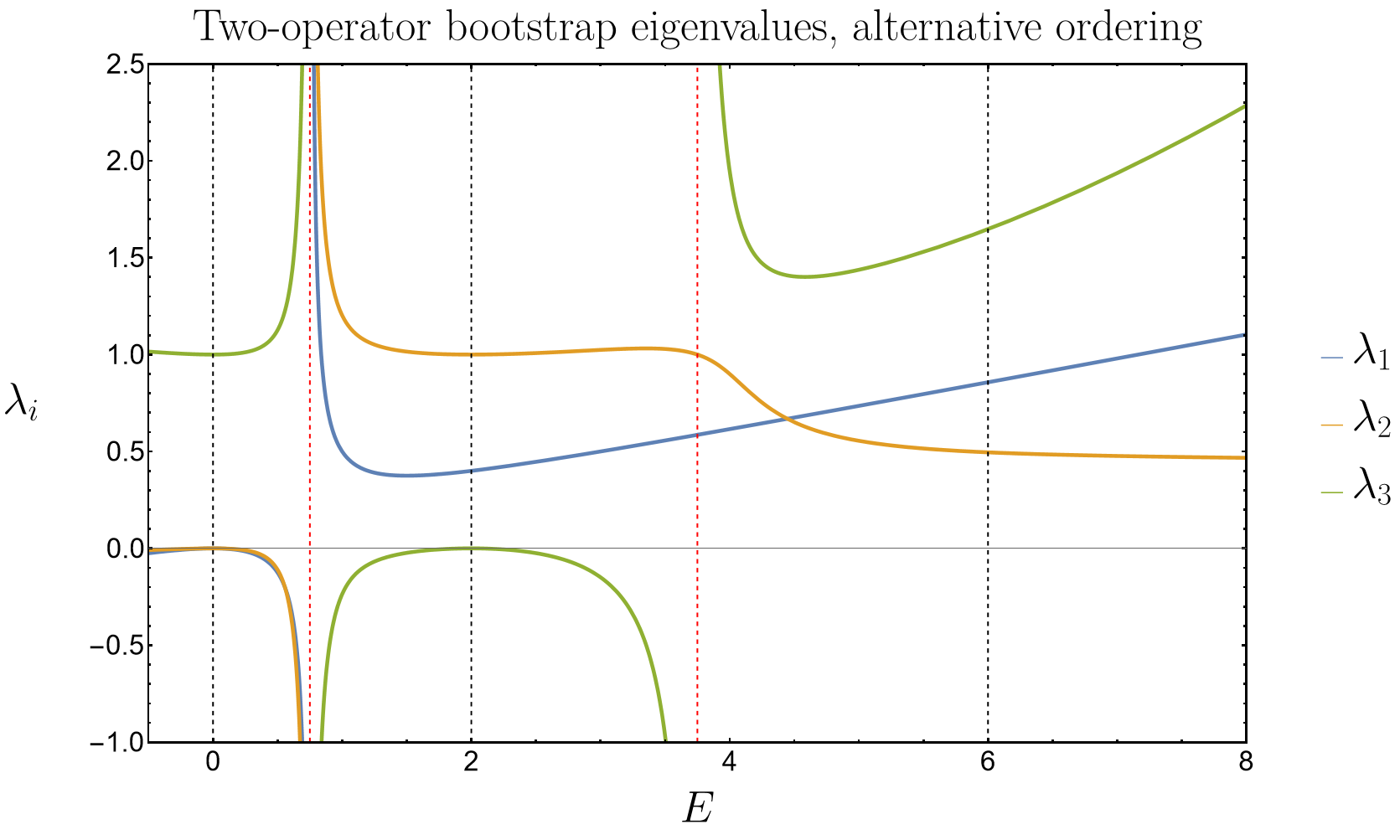}
\caption{Matrix eigenvalues $\lambda_{i}$ vs. energy $E$, for the alternatively ordered bootstrap matrix, $\tilde{\cB}_{2d}'$, at $K = 3$. Only at the precise energy levels $E=0,2$ and in band $E>15/4$ are all matrix eigenvalues non-negative. Red dashed lines are singularities in the $\tilde{\cB}_{2d}'$ eigenvalues and black dashed lines are the energy eigenvalues. }
\label{fig:tilb2dp_eig}
\end{center}
\end{figure}

\subsection{Obtaining exact energy levels}
\label{ss:eigen_analysis}
Here we study the eigenvalues of the $\tilde{\cB}_{2d}'$ matrix, to help understand the emergent exactness of $E$. Positive semi-definiteness is equivalent to its eigenvalues, $\lambda_{i}$, satisfying $\lambda_{i} \geq 0 \,,$ for all $i=1,\dots, K$. We compute $\lambda_{i}$ directly using the characteristic equation, $\det{\left(\tilde{\cB}_{2d}' - \lambda \mathbb{I}\right)} = 0$, and solve the constraints analytically. Using $K=3$ as an example, the solutions of these positivity constraints are $E=0$, $E=2$ and $E>15/4$. The product of the $\tilde{\cB}_{2d}'$ matrix's (see \eqref{eqn:boot_tilb2dp}) eigenvalues is captured by the determinant
\begin{equation}
\label{eqn:num_res_det_eigs}
\det{\tilde{\cB}_{2d}'} = \frac{E^4 (E-2)^2 (4E+21)}{16 (4E-3)^3 (4E-15)}\,,
\end{equation} and to provide a visual overview of the analytic results, we plot the corresponding eigenvalues vs. energy in Figure \ref{fig:tilb2dp_eig}. In the interval $E\in[-1,15/4]$ there is always at least one negative eigenvalue, except at $E=0,2$. These are precisely the first two energy levels of the system. This coincides with the findings in Figure \ref{fig:tilb2dp_step}, which were obtained numerically. We have checked the determinant of this matrix up to $K=12$. As in \eqref{eqn:num_res_det_eigs}, the energy eigenvalues of the system appear in the numerator of the determinant.

We tested the positivity of both $\tilde{\cB}_{2d}'$ and $\tilde{\cB}_{2d}''$ directly, up to $K=5$. The exact solutions to these constraint equations are the energy eigenvalues $E=0,2,6,12$ in both cases. The final single band behaviour of $\tilde{\cB}_{2d}'$ and the multiple-band behaviour for the $\tilde{\cB}_{2d}''$ matrix are also solutions, in agreement with the step search. In summary, for $\tilde{\cB}_{2d}$ at the finite $K$ considered, the first $K-1$ energy eigenvalues are fixed exactly by solving constraints $\lambda_{i} \geq 0 \,,$ for all $i=1,\dots, K$. While a general proof remains to be found, the analysis hints that this result should hold for arbitrary $K$.

%%%%%%%%%%%%%%%%%%%%%%%%%%%%%%%%%%%%%%%%%%%%%%%%%%%%%%%%%%%%%%%%%%%%%%%%%%%%%%%%%%%%%%%%%%%%%%%%%%%%%%%%%%SECTION%%%%%%%%%%%%%%%%%%%%%%%%%%%%%%%%%%%%%%%%%%%%%%%%%%%%%%%%%%%%%%%%%%%%%%%%%%%%%%%%%%%%%%%%%%%%%%%%%%%%%%%%%%%%%%%%%
\clearpage
\section{Conclusion}
\label{s:conclusion}

In this article, we have explored a quantum mechanical model defined on an interval and have shown that the bootstrap is capable of fixing its energy eigenvalues exactly. This adds to a number of examples in the literature, see \cite{Han:2020bkb,Lin:2020mme,Aikawa:2021qbl,Berenstein:2021loy,Khan:2022uyz,Du:2021hfw,Berenstein:2021dyf,Nancarrow:2022wdr,Berenstein:2022ygg,Berenstein:2022unr,Bhattacharya:2021btd,Hu:2022keu}, where numerical approximations constrained the expectation values. We began by constructing a self-adjoint Hamiltonian, $H=SZ(1-Z)S$, where $S$ and $Z$ are operators with $[S,Z]=i$. We found a set of 2d recursion relations on moments $\langle S^{\sigma}Z^{\zeta} \rangle$. The system is solvable which proves beneficial in showing the recursion relations are anomaly-free. A set of bootstrap matrices were then constructed using these relations, by considering the positivity of $\langle \cO^{\dagger} \cO \rangle$ and $\langle \cO \cO^{\dagger} \rangle$ where bootstrap operator $\cO$ is a linear combination of operators $Z^{\zeta}$, or composite $ S^{\sigma}[Z(1-Z)]^{\zeta}$. In the first case, we denote these bootstrap operators by $\cO_{1d}$ and in the second by $\cO_{2d}$.

Importantly, $S$ is not self-adjoint, therefore the calculation of $\langle \cO_{2d} \cO_{2d}^{\dagger} \rangle$ requires some care. In a particular octant of the $(\sigma, \zeta)$-plane, see the dark green region in Figure \ref{fig:octant_with_eqns}, the {\it Dagger lemma} \eqref{eqn:boot_lemma_rem_dag_eqn} guarantees that $S$ can be treated as a self-adjoint operator. Outside of this octant, we are forced to insert a complete set of states to evaluate the expectation values. The result can be a divergent sum, which after regularisation, depends on the position of the insertion, thereby making the $\langle \cO_{2d} \cO_{2d}^{\dagger} \rangle$ ill-defined, see Appendix \ref{app:fme_regularise}.

The positivity of $\langle \cO_{1d}^{\dagger} \cO_{1d} \rangle$ confined the possible energy eigenvalues into bands. We note that as in \cite{Berenstein:2022ygg}, one can consider the positivity of $\langle Z^{n}(1-Z)^{m} \cO_{1d}^{\dagger} \cO_{1d} \rangle$. We have tested $m,n = 0,1$ and did not find qualitatively better results. The positivity constraints of $\langle \cO_{2d}^{\dagger} \cO_{2d} \rangle$ performed poorly, unable to confine the eigenvalues as strongly as the 1d case. Quite unexpectedly, the positivity of the alternatively ordered $\langle \cO_{2d} \cO_{2d}^{\dagger} \rangle$ was able to identify an increasing number of energy eigenvalues exactly, for increasing matrix size.

The outcome of the bootstrap depends heavily on the choice and ordering of the bootstrap matrix operators. Understanding which choices are optimal would be beneficial. It is possible this particular model is special and therefore perturbations of the system could provide more intuition about the bootstrap. Surprisingly, the $\langle S Z \rangle$ moment does not feature in the bootstrap, dropping out of the positivity calculations. Understanding these cancellations may provide insight into why the bootstrap produces exact results.

A shortcoming of the current paper is the calculation of anomalies, where we used the analytic solution to show that they vanish. This may not be true for the general boundary conditions of \S\ref{ss:gen_bound_con} and such considerations are left for future work. Another possible research avenue is to better understand the finiteness argument in Appendix \ref{app:fme_main_proof}, potentially utilising the underlying supersymmetry of the system.

Using the coordinate transformation in \S\ref{ss:po_tell_coords}, it is possible to recast the calculations in terms of canonically conjugate coordinates $p$ and $u \equiv i\partial_{p}$. One can then consider bootstrap operators of form $\cO = u^{m}\sech^{n}(p/2) \tanh^{k}(p/2)$, for integers $m,n$ and $k=0,1$ and repeat the same tests for positivity. This may lead to interesting results.

In closing, we have shown that the bootstrap is able to identify the energy eigenvalues of the system exactly upon applying a finite number of positivity constraints. We look forward to revisiting the issues and challenges discussed in this section in our future endeavours.

{ \bf Acknowledgements}: We thank Laurentiu Rodina and Sean Hartnoll for useful discussions. We thank Yuan Xin for correspondence and David Berenstein for insightful comments on the paper. LS is
supported by an STFC quota studentship. DV is supported by the STFC Ernest Rutherford
grant ST/P004334/1 and by the STFC Consolidated
Grant ST/T000686/1 ``Amplitudes, strings \& duality". No new data were generated or analysed during this study.

%%%%%%%%%%%%%%%%%%%%%%%%%%%%%%%%%%%%%%%%%%%%%%%%%%%%%%%%%%%%%%%%%%%%SECTION%%%%%%%%%%%%%%%%%%%%%%%%%%%%%%%%%%%%%%%%%%%%%%%%%%%%%%%%%%%%%%%
\newpage
\clearpage

\appendix

\section{Alternative derivation of the \texorpdfstring{$\langle Z^{\zeta} \rangle$}{} recursion relation }
\label{app:drr}
The recursion relation in $\langle Z^{\zeta} \rangle$ moments can be alternatively derived as follows\footnote{Note, the result from deriving this recursion relation is valid provided there are no anomalies, as discussed in \S \ref{ss:anomalies}}. To begin, we restate the Hamiltonian 
\begin{equation}
H = SZ(1-Z)S = S^{2} Z(1-Z) + i S(2Z-1)\,,
\end{equation}
the operator $S= i \partial_{z}$ and the commutator $[S,Z] = i$. We introduce three operators: $\cO_{a}(Z)$, $\cO_{c}(S,Z) = S\cO_{b}(Z)$ and $\cO_{d}(Z)$ that feature in two commutation equations and an energy equation 
\begin{equation}
\label{eqn:app_drr_oa}
\langle [H, \cO_{a}(Z)]\rangle =0 \,,
\end{equation}
\begin{equation}
\label{eqn:app_drr_oc}
\langle [H, \cO_{c}(S,Z)]\rangle =0 \,,
\end{equation}
\begin{equation}
\label{eqn:app_drr_h_od}
\langle  H \cO_{d}(Z)\rangle = E \langle \cO_{d}(Z) \rangle \,.
\end{equation}
The motivation for this initial setup is based on similar calculations seen in the literature, for example \cite{Han:2020bkb,Berenstein:2021dyf}. Starting with equation \eqref{eqn:app_drr_oa}, using the commutator to order all $S$ operators to the left, we obtain
\begin{equation}
\label{eqn:app_drr_oa_full_exp}
\langle [H, \cO_{a}(Z)]\rangle = 2\langle S[S,\cO_{a}(Z)]Z(1-Z) \rangle + \langle [[S,\cO_{a}(Z)],S]Z(1-Z)\rangle + i \langle [S,\cO_{a}(Z)](2Z-1)\rangle = 0 \,.
\end{equation}
Similarly for equation \eqref{eqn:app_drr_oc}
\begin{equation}
\label{eqn:app_drr_oc_full_exp}
\langle [H, \cO_{c}(S,Z)]\rangle = \langle S^2 \alpha_{1}(Z) \rangle + \langle S \alpha_{2}(Z) \rangle = 0\,,
\end{equation}
where
\begin{equation}
\alpha_{1}(Z) = 2[S,\cO_{b}(Z)]Z(1-Z) + [Z(1-Z),S]\cO_{b}(Z) \,,
\end{equation}
\begin{equation}
\alpha_{2}(Z) = [[S,\cO_{b}(Z)],S]Z(1-Z) + i [2Z-1,S]\cO_{b}(Z) + i [S, \cO_{b}(Z)](2Z-1)  \,.
\end{equation}
As we want to create a recursion relation in $Z$ alone, we must eliminate the expectation values of form $\langle S^2  h(Z) \rangle$ and $\langle S g(Z) \rangle$, for the arbitrary functions $h(Z),g(Z)$ that may appear. By defining
\begin{equation}
\label{eqn:app_drr_od}
\cO_{d}(Z) = (Z(1-Z))^{-1} \alpha_{1}(Z) \,,
\end{equation}
we can insert this into \eqref{eqn:app_drr_h_od} to obtain
\begin{align}
\label{eqn:app_drr_od_full_exp}
\langle S^2 Z(1-Z)\cO_{d}(Z) \rangle + i\langle S (2Z-1)  \cO_{d}(Z)  \rangle &= E\langle \cO_{d}(Z) \rangle
\end{align}
which, after expanding out, becomes
\begin{equation}
\label{eqn:app_drr_s_squared_alph_1}
\langle S^2 \alpha_{1}(Z) \rangle + i\langle S (2Z-1)  (Z(1-Z))^{-1} \alpha_{1}(Z)  \rangle
= E\langle (Z(1-Z))^{-1}\alpha_{1}(Z) \rangle  \,.
\end{equation}
We now substitute equation \eqref{eqn:app_drr_s_squared_alph_1} into  \eqref{eqn:app_drr_oc_full_exp}, to remove the $\langle S^{2} \alpha_{1}(Z) \rangle$ term,
\begin{equation}
\label{eqn:app_drr_od_full_exp2}
E\langle (Z(1-Z))^{-1}\alpha_{1}(Z) \rangle + \langle S \left(\alpha_{2}(Z) - i (2Z-1)  (Z(1-Z))^{-1} \alpha_{1}(Z)  \right)\rangle = 0
\end{equation}
where we have grouped the terms that $S$ acts on. We are left with two equations, \eqref{eqn:app_drr_oa_full_exp} and \eqref{eqn:app_drr_od_full_exp2}, containing terms of $\langle S g(Z) \rangle$ form. To eliminate these terms, we look to substitute one equation into the other and can do so, providing that the operators  $S$ acts on in both equations are equal. Hence
\begin{equation}
\label{app_drr_diff_eqn}
 2 [S,\cO_{a}(Z)]Z(1-Z) =  \alpha_{2}(Z) - i(2Z-1)(Z(1-Z))^{-1} \alpha_{1}(Z)  \,.
\end{equation}
Acting on an arbitrary wavefunction with each side of this operator equation, produces a differential equation in functions $\cO_{a}(z)$ and $\cO_{b}(z)$\footnote{Note that $\alpha_{1}(Z), \alpha_{2}(Z)$ are functions of the operator $\cO_{b}(Z)$.}. This equation is
\begin{equation}
\label{eqn:app_drr_diff_eqn_oa_ob}
\left(\frac{1}{z(1-z)}-2 \right)\cO_{b}(z) + (2z-1) \cO_{b}'(z)+z(1-z)\cO_{b}''(z) = 2 i z(1-z) \cO_{a}'(z)
\end{equation}
with $\cO'(z) = \partial_{z}\cO(z)$. Setting\footnote{This choice of $\cO_{b}$, leads to a final recursion relation that is expressed in terms of $\langle Z^{\zeta} \rangle$-type moments. The motivation for this form of $\cO_{b}$ was an educated guess, based on the form of the Hamiltonian (containing $Z(1-Z)$) as well as the aim to relate $\langle Z^{\zeta} \rangle$-type moments.}
\begin{equation}
\cO_{b}(z)= z^{\zeta}(1-z)
\end{equation}
produces a differential equation in $\cO_{a}(z)$, with solution
\begin{equation}
\cO_{a}(z) = -\frac{i}{2}(\zeta-1)z^{\zeta-1}(1-z) \,.
\end{equation} where the integration constant has been set to zero. We then promote these functions of coordinate $z$, back to functions of operators
\begin{equation}
\label{eqn:app_drr_ob_oa_ops}
\cO_{b}(Z)= Z^{\zeta}(1-Z) \,, \qquad \cO_{a}(Z) = -\frac{i}{2}(\zeta-1)Z^{\zeta-1}(1-Z) \,.
\end{equation}
Therefore, under these choices, the operators multiplying $S$ in \eqref{eqn:app_drr_oa_full_exp} and \eqref{eqn:app_drr_od_full_exp2} are equal and we proceed to substitute \eqref{eqn:app_drr_od_full_exp2} into \eqref{eqn:app_drr_oa_full_exp} giving
\begin{equation}
\label{eqn:app_drr_final_rec_in_s_z}
-E\langle (Z(1-Z))^{-1} \alpha_{1}(Z) \rangle + \langle [[S,\cO_{a}(Z)],S]Z(1-Z)\rangle + i \langle [S,\cO_{a}(Z)](2Z-1)  \rangle = 0 \,.
\end{equation}
Inserting $\cO_{a}(Z)$ and $\cO_{b}(Z)$ from \eqref{eqn:app_drr_ob_oa_ops} and evaluating the commutators reveals the final recursion relation
\begin{equation}
\label{eqn:app_drr_final_rec}
\text{\eqnverty } : \quad (\zeta-1)^3 f_{\zeta - 2} +  (2 \zeta-1) \left(2 E - \zeta^2 + \zeta \right) f_{\zeta-1} + \zeta \left(\zeta^2 -4 E  - 1 \right) f_{\zeta}=0\,,
\end{equation}
where $f_{\zeta} \equiv f_{0,\zeta} = \langle Z^{\zeta} \rangle$, with general 2d moments defined as $f_{\sigma, \zeta} = \langle S^{\sigma}Z^{\zeta} \rangle$.
This agrees with the original result in equation \eqref{eqn:boot_1d_rec_in_z} of \S\ref{ss:rec_rels}, and the process outlined here can also be repeated to identify the \eqnhoriy recursion relation exclusively in $\langle S^{\sigma} \rangle $ moments, seen in \eqref{eqn:boot_1d_in_s_hori}.

%%%%%%%%%%%%%%%%%%%%%%%%%%%%%%%%%%%%%%%%%%%%%%%%%%%%%%%%%%%%%%%%%%%%SECTION%%%%%%%%%%%%%%%%%%%%%%%%%%%%%%%%%%%%%%%%%%%%%%%%%%%%%%%%%%%%%%%
\newpage
\section{Proof of Dagger Lemma}
\label{app:lemma_proof_dag_remove}
In this appendix we provide the proof of the Dagger lemma introduced in \S\ref{sss:two_op_matrix},
\begin{namedtheorem}[``Dagger"]
For $\beta \geq \alpha \geq 0$, the following is true
\begin{equation}
\label{eqn:app_lemma_rem_dag_eqn}
[Z(1-Z)]^{\beta}(S^{\alpha})^{\dagger} = [Z(1-Z)]^{\beta}S^{\alpha} \,.
\end{equation}
\end{namedtheorem}

\begin{proof}
For $\alpha =0$, the statement is trivial so the following considers $\alpha \geq 1$. Take $\phi$ and $\psi$ as arbitrary wavefunctions and calculate
{\small
\begin{equation}
\label{eqn:herm_cond}
\langle S^{\alpha}[Z(1-Z)]^{\beta} \phi| \psi \rangle - \langle \phi| [Z(1-Z)]^{\beta}S^{\alpha}\psi \rangle = \int_{0}^{1}  \left( \left[S^{\alpha}[Z(1-Z)]^{\beta} \phi(z) \right]^{*} - \phi(z)^{*} [Z(1-Z)]^{\beta}S^{\alpha} \right)\psi(z) dz \,,
\end{equation}}Then when \eqref{eqn:herm_cond} vanishes, it implies equation \eqref{eqn:app_lemma_rem_dag_eqn}.
Integrating \eqref{eqn:herm_cond} by parts, and equating to zero gives
\begin{equation}
\label{eqn:boot_bc_bootmat_2d}
-i^\alpha \sum_{\kappa=0}^{\alpha-1} (-1)^\kappa \frac{\partial^\kappa}{\partial z^\kappa} \left[\phi(z)^{*} [z(1-z)]^\beta \right] \cdot \frac{\partial^{\alpha-1-\kappa}}{\partial z^{\alpha-1-\kappa}} \psi(z) \bigg|_{0}^{1} =0 \,.
\end{equation}
The aim is to show that, for an appropriate choice of $\alpha$ and $\beta$, each term in the above sum contains sufficient factors of $z(1-z)$ such that when we apply the boundary conditions (see equations \eqref{eqn:model_bcs} and \eqref{eqn:model_bc_psi}), they vanish.  To start, we set $\partial_{z}^{(\alpha)}\psi(z) \equiv \frac{\partial^{\alpha}}{\partial z^{\alpha}}\psi(z)$ and apply the general Leibniz rule\footnote{Defined by $
\partial_{z}^{(\kappa)} \left(\phi(z)^{*} [z(1-z)]^{\beta} \right) = \sum_{\lambda=0}^{\kappa} \binom{\kappa}{\lambda} \partial_{z}^{(\kappa-\lambda)}\phi(z)^{*} \partial_{z}^{(\lambda)}[z(1-z)]^{\beta} \,.$}
to rewrite equation \eqref{eqn:boot_bc_bootmat_2d} as
\begin{equation}
\label{eqn:boot_bc_bootmat_2d_expanded}
-i^\alpha \sum_{\kappa=0}^{\alpha-1} \sum_{\lambda=0}^{\kappa} (-1)^\kappa \binom{\kappa}{\lambda} \partial^{(\kappa-\lambda)}_{z} \phi(z)^{*} \cdot  \partial_{z}^{(\lambda)} [z(1-z)]^\beta \cdot \partial^{(\alpha-1-\kappa)}_{z} \psi(z) \bigg|_{0}^{1} =0 \,.
\end{equation}
Factors of $z(1-z)$ are introduced either by the derivatives of the wavefunctions $\psi(z)$ and $\phi(z)^{*}$, or by derivatives of $[z(1-z)]^\beta$. $\partial^{(\alpha -1 - \kappa)}_{z}\psi(z)$ and $\partial^{(\kappa-\lambda)}_{z}\phi(z)^{*}$ can be written in terms of zeroth and first derivatives by applying multiple derivatives and substitutions of the Schr{\"o}dinger equation. We have that 
\begin{equation}
\partial_{z}^{(2)}\psi(z) \equiv \psi''(z) = [z(1-z)]^{-1} ((2z-1)\psi'(z) -E\psi(z))\,,
\end{equation}
and hence each subsequent derivative introduces another $[z(1-z)]^{-1}$ factor, e.g. for $\alpha - 1 - \kappa \geq 2$, $\partial_{z}^{(\alpha - 1 - \kappa)}\psi(z)$ introduces $[z(1-z)]^{-\alpha + \kappa +2}$. Therefore, in order to apply the boundary conditions, the exponent of the $z(1-z)$ factor in the summand should be the sum of the derivative exponents\footnote{This ensures there are enough $z(1-z)$ factors that even the $\phi'(z)^{*} \psi'(z)$ terms that appear, vanish under application of the boundary condition.}: $(\alpha-1-\kappa) + (\kappa-\lambda) = \alpha -1-\lambda$. We already have a contribution of $[z(1-z)]^{\beta-\lambda}$ from the $\partial_{z}^{(\lambda)}[z(1-z)]^{\beta}$ term, hence $ \beta - \lambda \geq \alpha -1 - \lambda \implies \beta \geq \alpha-1$. We note that this counting argument is only sufficient when there is at least one derivative present on $\psi$ or $\phi^{*}$. For the term in the summand that has no wavefunction derivatives, the sum of derivative exponents is zero and hence the previous counting argument implies $\beta -\lambda\geq 0$. However, this fails at $\beta-\lambda =0$ as the wavefunctions do not necessarily vanish at the boundaries. To ensure the summand term does vanish, we must introduce an additional $z(1-z)$ factor which leads to the strongest constraint: $\beta \geq \alpha$. By taking $\beta \geq \alpha$, every term in \eqref{eqn:boot_bc_bootmat_2d_expanded} individually vanishes and therefore the lemma equation  \eqref{eqn:app_lemma_rem_dag_eqn} is satisfied.
\end{proof}

\newpage

%%%%%%%%%%%%%%%%%%%%%%%%%%%%%%%%%%%%%%%%%%%%%%%%%%%%%%%%%%%%%%%%%%%%%%%%%%%%%%%%%%
%%%%%%%%%%%%%%%%%%%%%%%%%%SECTION%%%%%%%%%%%%%%%%%%%%%%%%%%%%%%%%%%%%%%%
%%%%%%%%%%%%%%%%%%%%%%%%%%%%%%%%%%%%%%%%%%%%%%%%%%%%%%%%%%%%%%%%%%%%%%%%%%%%%%%%%%

\section{Finiteness of matrix elements}
\label{app:fme_main_proof}

This section analyses the finiteness of 2d bootstrap matrix elements. This involves placing restrictions on the indices of expectation values of form $\langle S^{\sigma}[Z(1-Z)]^{\zeta} \rangle$. These restrictions also apply to the $f_{\sigma,\zeta} = \langle S^{\sigma}Z^{\zeta}\rangle$ moments since these comprise the matrix elements, see equation \eqref{eqn:boot_oodag_ele_expanded_mccoy}.

The 2d bootstrap matrix $\tilde{\cB}_{2d}$ consists of elements
\begin{align}
\begin{split}
\label{eqn:app_tildb2d_mat_ele_oper}
\left(\tilde{\cB}_{2d}\right)_{(\alpha,\beta),(\alpha',\beta')}
&= \langle S^{\alpha}[Z(1-Z)]^{\beta+\beta'} \left( S^{\alpha'} \right)^{\dagger} \rangle
\end{split}
\end{align}
and here we show that upon taking $\beta \geq \alpha$ and $\beta' \geq \alpha'$, such elements are finite. Note that throughout this section, we explicitly use the energy eigenfunctions $\psi_{a} = \sqrt{2a+1}P_{a}(2z-1)=\sqrt{2a+1}\tilde{P}_{a}(z)$. This is to emphasise that this proof relies on the specific analytic form of the solutions.

To begin, consider the $\tilde{\cB}_{2d}$ bootstrap matrix element, evaluated on eigenfunction $\psi_{a}$
\begin{align}
\label{eqn:app_fn_b2d_mat_ele}
\begin{split}
\left(\tilde{\cB}_{2d}\right)_{(\alpha,\beta),(\alpha',\beta')} &= \bra{\psi_{a}} S^{\alpha} [Z(1-Z)]^{\beta + \beta'} (S^{\alpha'})^{\dagger} \ket{\psi_{a}}
\\
&= \sum_{c=0}^{\infty}  \bra{\psi_{a}} S^{\alpha} [Z(1-Z)]^{\beta} \ket{\psi_{c}}\bra{\psi_{c}} [Z(1-Z)]^{\beta'} (S^{\alpha'})^{\dagger} \ket{\psi_{a}}
\end{split}
\end{align}
where a complete set of energy eigenstates $\sum_{c=0}^{\infty}\ket{\psi_{c}}\bra{\psi_{c}}$ has been inserted centrally. To show that \eqref{eqn:app_fn_b2d_mat_ele} is finite, it suffices to show that at a chosen $a , \alpha, \beta, \alpha', \beta' $, the vectors indexed by $c$,
\begin{equation}
\label{eqn:app_fn_two_vecs_indexed_by_c}
\cI_{c}^{(a,\alpha,\beta)} = \bra{\psi_{a}} S^{\alpha} [Z(1-Z)]^{\beta} \ket{\psi_{c}} \,, \qquad \bar{\cI}_{c}^{(a,\alpha',\beta')} = \bra{\psi_{c}}  [Z(1-Z)]^{\beta'} (S^{\alpha'})^{\dagger} \ket{\psi_{a}}
\end{equation}
only have a finite number of non-zero components. Noting that since $\bar{\cI}_{c}^{a, \alpha',\beta'} = (\cI_{c}^{a,\alpha',\beta'})^{*}$, then providing we apply the same condition between $\alpha'$ and $\beta'$ as we do between $\alpha$ and $\beta$, then we need only show $\cI^{(a,\alpha,\beta)}_{c}$ is finite. We start by exchanging the derivatives onto the $\psi_{a}(z)$ using repeated integration by parts
\begin{align}
\mathcal{I}_{c}^{(a,\alpha,\beta)} &= \bra{\psi_{a}} S^{\alpha} [Z(1-Z)]^{\beta} \ket{\psi_{c}}
\\
&= i^{\alpha} \int^{1}_{0} \psi_{a}(z) \partial_{z}^{(\alpha)}\left( [z(1-z)]^{\beta} \psi_{c}(z)\right) dz
\\
\begin{split}
\label{eqn:app_fn_bt_and_int_cvec}
&= i^{\alpha} \sum_{k=0}^{\alpha-1} (-1)^{k} (\partial_{z}^{(k)}\psi_{a}(z)) \partial_{z}^{(\alpha-1-k)} \left([z(1-z)]^{\beta} \psi_{c}(z) \right) \Bigg|_{0}^{1} 
\\
& \qquad \qquad \qquad \qquad \qquad \qquad   + (-i)^{\alpha} \int_{0}^{1} (\partial_{z}^{(\alpha)}\psi_{a}(z)) [z(1-z)]^{\beta} \psi_{c}(z) dz
\end{split}
\\
\label{eqn:app_fn_deriv_on_state_a}
&=(-i)^{\alpha} \int_{0}^{1} (\partial_{z}^{(\alpha)}\psi_{a}(z)) [z(1-z)]^{\beta} \psi_{c}(z) dz
\\
\label{eqn:app_fn_deriv_three_terms}
&= (-i)^{\alpha}\sqrt{(2a+1)(2c+1)} \int_{0}^{1} \left(\partial_{z}^{(\alpha)} \tilde{P}_{a}(z)\right) [z(1-z)]^{\beta} \tilde{P}_{c}(z) dz
\end{align}
where again, $\tilde{P}_{a}(z) = P_{a}(2z-1)$. Importantly, since $\psi_{a}$ and $\psi_{c}$ are polynomials, the boundary term sum in \eqref{eqn:app_fn_bt_and_int_cvec} will vanish, providing that $\beta \geq \alpha$. This inequality ensures there are sufficient factors of $z(1-z)$ in each term so that when evaluated at $z = 0,1$, they all vanish individually\footnote{Note, for $\alpha = 0$, we do not require this process -- there are no derivatives, so no boundary terms and the integrand will simply be a well-defined polynomial in $z$.}. Binomially expanding $(1-z)^{\beta}$ gives
\begin{equation}
[z(1-z)]^{\beta}  = \sum_{k=0}^{\beta} (-1)^{k} \binom{\beta}{k} z^{\beta + k} \,,
\end{equation}
and by using
\begin{equation}
\tilde{P}_{1}(z) = 2z-1 \implies z = \frac{1}{2}(1+\tilde{P}_{1}(z)) \,,
\end{equation}
we have
\begin{equation}
\label{eqn:css_pol_z_one_min_z_to_beta}
[z(1-z)]^{\beta} = \sum_{k=0}^{\beta} (-1)^{k} \binom{\beta}{k} \left[\frac{1}{2}(1+\tilde{P}_{1}(z))\right]^{\beta + k} = \sum_{k=0}^{\beta} \sum_{\ell=0}^{\beta+k} (-1)^{k}\left(\frac{1}{2} \right)^{\beta + k} \binom{\beta}{k} \binom{\beta +k}{\ell} (\tilde{P}_{1}(z))^{\ell}  \,,
\end{equation}
where $(1+\tilde{P}_{1}(z))^{\beta + k}$ has also been expanded. To evaluate $(\tilde{P}_{1}(z))^{\ell}$ we use the Legendre polynomial product rule \cite{adams_legendre,donnel_trip_leg_p}
\begin{equation}
\label{eqn:css_pol_tilde_prod}
\tilde{P}_{m}(z) \tilde{P}_{n}(z) = \sum_{\ell=|m-n|}^{m+n} \begin{pmatrix}
 m & n & \ell \\
0 & 0 & 0
\end{pmatrix}^2 (2\ell+1) \tilde{P}_{\ell}(z) = \sum_{\ell=|m-n|}^{m+n}  (w_{m,n,\ell})^2 (2\ell+1) \tilde{P}_{\ell}(z)
\end{equation}
where
\begin{equation}
w_{m,n,\ell} := \begin{pmatrix}
 m & n & \ell \\
0 & 0 & 0
\end{pmatrix} = \left[\int_{0}^{1} \tilde{P}_{m}(z)\tilde{P}_{n}(z)\tilde{P}_{\ell}(z)dz \right]^{\frac{1}{2}} \,,
\end{equation}
is the Wigner (3j) coefficient. Applying the rule multiple times, gives general formula
\begin{equation}
\label{eqn:css_pol_multiple_p1}
(\tilde{P}_{1}(z))^{\ell} =
\begin{cases}
\tilde{P}_{0}(z), &  \text{ for } \ell = 0 \\
\tilde{P}_{1}(z), &  \text{ for } \ell = 1 \\
\prod_{r=2}^{\ell} \left(\sum_{i_{r}=|i_{r-1}-1|}^{i_{r-1}+1} (w_{1,i_{r-1},i_{r}})^2 (2i_{r}+1) \right) \tilde{P}_{i_{\ell}}(z) \quad \text{with} \quad i_{1}=1 \,, \quad & \text{ for } \ell \geq 2 \,.
\end{cases}
\end{equation}
Note that $(\tilde{P}_{1}(z))^{0} =1 = \tilde{P}_{0}(z)$. So, we can cast $[z(1-z)]^{\beta} \tilde{P}_{c}(z)$ from the  \eqref{eqn:app_fn_deriv_three_terms} integrand in terms of multiple sums, over the product of two Legendre polynomials by substituting \eqref{eqn:css_pol_multiple_p1} into  \eqref{eqn:css_pol_z_one_min_z_to_beta} and multiplying by $\tilde{P}_{c}(z)$. The last term in the integrand to consider is $\partial_{z}^{(\alpha)} \tilde{P}_{a}(z)$ which can be evaluated\footnote{Note the factor of $2^{\alpha}$ appearing in \eqref{eqn:css_deriv_formula} due to the fact $\tilde{P}_{a}(z)$ is used instead of $P_{a}(z)$.} (see \cite{laurent2017scaling}) as
\begin{equation}
\label{eqn:css_deriv_formula}
\partial_{z}^{(\alpha)} \tilde{P}_{a}(z) \equiv \frac{d^{\alpha}}{dz^{\alpha}}\tilde{P}_{a}(z) = 2^{\alpha} \sum_{m=0}^{\lfloor (a-\alpha)/2 \rfloor} \gamma_{a-\alpha-2m} \tilde{P}_{a-\alpha-2m}(z)
\end{equation}
with the floor expression, $\lfloor (a-\alpha)/2 \rfloor$, meaning the largest integer less than or equal to $(a-\alpha)/2$. $\gamma_{a-\alpha-2m}$ is a recursion coefficient and can be obtained from formula
\begin{align}
\begin{split}
\label{eqn:css_recur_coef_for_deriv}
\gamma_{a-\alpha -2m} &= \frac{2^{\alpha + 2m}(a-\frac{1}{2})^{\underline{\alpha}}(a - m)^{\underline{m}}(a-\alpha -\frac{1}{2})^{\underline{2m}}}{(2m)^{\underline{2m}} (a - \frac{1}{2})^{\underline{m}}}
- \sum_{j=0}^{m-1} \frac{ (2(a - \alpha - m - j))^{\underline{2(m-j)}} }{(2(m-j))^{\underline{2(m-j)}}} \gamma_{a-\alpha-2j}
\end{split}
\end{align}
where $(x)^{\underline{n}}$ represents the falling factorial. Combining these component expressions together, the full formula for $\mathcal{I}_{c}^{(a,\alpha,\beta)}$ may be written as
\begin{align}
\begin{split}
\label{eqn:css_pol_full_formula}
&\cI_{c}^{(a,\alpha,\beta)} = (-i)^{\alpha} \sqrt{(2a+1)(2c+1)} \sum_{k=0}^{\beta} \sum_{\ell=0}^{\beta + k} (-1)^{k} \left(\frac{1}{2}\right)^{\beta+k} \binom{\beta}{k} \binom{\beta +k}{\ell} \times
\\
&  2^{\alpha}  \sum_{m=0}^{\lfloor (a-\alpha)/2 \rfloor}\gamma_{a-\alpha-2m} \int_{0}^{1} \tilde{P}_{a-\alpha-2m}(z) \Bigg[ \delta_{\ell,0} \tilde{P}_{0}(z) + \delta_{\ell,1} \tilde{P}_{1}(z) 
\\
& \qquad \quad +(1-\delta_{\ell,0})(1-\delta_{\ell,1})  \prod_{r=2}^{\ell} \left(\sum_{i_{r}=|i_{r-1}-1|}^{i_{r-1}+1} w_{1,i_{r-1},i_{r}}^2 (2i_{r}+1) \right) \tilde{P}_{i_{\ell}}(z) \Bigg] \tilde{P}_{c}(z) dz \,,
\end{split}
\end{align}
where Kronecker deltas have been introduced according to \eqref{eqn:css_pol_multiple_p1}. The equation is a set of sums over $\{k,\ell,m, i_{2}, \dots, i_{\ell} \}$ and the key terms are the triple integrals of form
\begin{equation}
\label{eqn:css_trip_int_def}
\int_{0}^{1} \tilde{P}_{a-\alpha-2m}(z) \tilde{P}_{q}(z) \tilde{P}_{c}(z) dz = \begin{pmatrix}
a-\alpha-2m & q & c \\
0 & 0 & 0
\end{pmatrix}^2 \equiv (w_{a- \alpha - 2m, q, c})^2 \,,
\end{equation}
where $q \in \{0,1, i_{\ell} \}$. Using \eqref{eqn:css_trip_int_def} to exchange the integrals, the result is determined by the $(3j)$ coefficients
\begin{align}
\begin{split}
\label{eqn:css_pol_full_formula_pt2}
&\cI_{c}^{(a,\alpha,\beta)}= (-i)^{\alpha} \sqrt{(2a+1)(2c+1)} \sum_{k=0}^{\beta} \sum_{\ell=0}^{\beta + k} (-1)^{k} \left(\frac{1}{2}\right)^{\beta+k} \binom{\beta}{k} \binom{\beta +k}{\ell} \times
\\
&  2^{\alpha}  \sum_{m=0}^{\lfloor (a-\alpha)/2 \rfloor}\gamma_{a-\alpha-2m} \Bigg[ \delta_{\ell,0} (w_{a-\alpha-2m,0,c})^2 + \delta_{\ell,1} (w_{a-\alpha-2m,1,c})^2 
\\
& \qquad +(1-\delta_{\ell,0})(1-\delta_{\ell,1}) \prod_{r=2}^{\ell} \left(\sum_{i_{r}=|i_{r-1}-1|}^{i_{r-1}+1} w_{1,i_{r-1},i_{r}}^2 (2i_{r}+1) \right)  (w_{a-\alpha-2m,i_{\ell},c})^2  \Bigg] \,.
\end{split}
\end{align}

Importantly, since $a-\alpha -2m$ and $q$ can only take finite values (that is, these values are either chosen directly, or if they correspond to sum variables, their limits are fixed and finite), then $w_{a-\alpha-2m,q,c}$ is non-zero for only a finite number of values for $c$.  This is due to the Wigner $(3j)$ coefficient's triangle inequality property/selection rule: $w_{a-\alpha-2m,q,c} = 0$ unless
\begin{equation}
\label{eqn:css_wig_sel_rule}
|a-\alpha-2m - q | \leq c \leq a-\alpha-2m + q \,. 	
\end{equation}
In turn, $\mathcal{I}^{(a,\alpha,\beta)}$, as a $c$-component vector for selected $(a,\alpha,\beta)$ with $\beta \geq \alpha$, is restricted to have a finite number of non-zero components, even though $c$ runs from $0$ to $\infty$.  Hence, ensuring that $\beta \geq \alpha$ and $\beta' \geq \alpha'$ means matrix $\tilde{\cB}_{2d}$ consists of finite elements when the associated expectation values are evaluated using the energy eigenstates. When $\alpha = \alpha'$ and $\beta= \beta'$ we see that this also implies the state with index $c$ has a finite norm squared.

The above proof considered the alternatively ordered $\tilde{\cB}_{2d}$ and not $\cB_{2d}$. Following the same arguments, the elements of matrix $\cB_{2d}$ are finite with no additional relation between $\alpha$ and $\beta$ required. The reason for this is that in the $\cB_{2d}$ elements, the derivatives act directly on the chosen state $\psi_{a}$ and not on $\psi_{c}$. This means no boundary terms appear as there is no integration by parts. However, we still need to remove the $S^{\dagger}$ operators from the matrix elements to employ the bootstrap. For $\cB_{2d}$, it means we must take $\beta \geq \alpha$ according to the Dagger lemma of \S\ref{sss:two_op_matrix}, but leaves freedom to choose other $\alpha'$ and $\beta'$ in this case. Such choices are reserved for future investigation.

%%%%%%%%%%%%%%%%%%%%%%%%%%%%%%%%%%%%%%%%%%%%%%%%%%%%%%%%%%%%%%%%%%%%%%%%%%%%%%%%%%
%%%%%%%%%%%%%%%%%%%%%%%%%%SUBSECTION%%%%%%%%%%%%%%%%%%%%%%%%%%%%%%%%%%%%%%%
%%%%%%%%%%%%%%%%%%%%%%%%%%%%%%%%%%%%%%%%%%%%%%%%%%%%%%%%%%%%%%%%%%%%%%%%%%%%%%%%%%

\section{Upper and lower octant examples}
\label{app:fme_up_low_examples}
Here we provide explicit upper and lower octant examples of the $c$-index vectors discussed in Appendix \ref{app:fme_main_proof}. Let us concentrate on matrix elements $(\tilde{\cB}_{2d})_{(\alpha,\beta), (\alpha',\beta')}$ with $\alpha = \alpha'$, $\beta = \beta'$ for simplicity. Beginning with the lower octant, i.e. $\alpha > \beta$, set $\alpha = 2$, $\beta = 1$ such that
\begin{equation}
(\tilde{\cB}_{2d})_{(2,1), (2,1)} = \sum_{c=0}^{\infty} \braket{\psi_{a}| S^2 Z(1-Z) | \psi_{c}} \braket{\psi_{c}| Z(1-Z)(S^2)^{\dagger}| \psi_{a}} \,.
\end{equation}
To show this leads to an infinite matrix element, it suffices to show that 
\begin{equation}
\cI_{c}^{a,2,1} = \braket{\psi_{a}| S^2 Z(1-Z) | \psi_{c}} \,,
\end{equation}
has infinitely many non-zero components as a $c$-index vector. These components are evaluated using integration by parts twice
\begin{align}
\cI_{c}^{a,2,1} &= \braket{\psi_{a} | S^{2} Z(1-Z) | \psi_{c}}
\\
&=i^2 \sqrt{(2a+1)(2c+1)} \int_{0}^{1} \tilde{P}_{a}(z) \partial_{z}^{2} \left[ z(1-z) \tilde{P}_{c}(z) \right] dz
\\
\begin{split}
\label{eqn:css_inf_example_bt2}
&= -\sqrt{(2a+1)(2c+1)} \Bigg[ \tilde{P}_{a}(z) \partial_{z} [z(1-z) \tilde{P}_{c}(z)] \bigg|_{0}^{1} 
\\
&  \qquad \qquad \qquad - (\partial_{z} \tilde{P}_{a}(z)) z(1-z) \tilde{P}_{c}(z) \bigg|_{0}^{1} + \int_{0}^{1} (\partial_{z}^2 \tilde{P}_{a}(z)) z(1-z) \tilde{P}_{c}(z) dz \Bigg]
\end{split}
\\
\label{eqn:css_inf_example_bt1_remains}
&= -\sqrt{(2a+1)(2c+1)} \Bigg[ \tilde{P}_{a}(z) \partial_{z} [z(1-z) \tilde{P}_{c}(z)] \bigg|_{0}^{1} + \int_{0}^{1} (\partial_{z}^2 \tilde{P}_{a}(z)) z(1-z) \tilde{P}_{c}(z) dz \Bigg] \,.
\end{align}
The second boundary term that appears in \eqref{eqn:css_inf_example_bt2} will vanish for all states $a$ due to the $z(1-z)$ factor (using knowledge that $\tilde{P}_{a}(z)$ are polynomials for all $a$). However, the remaining boundary term and integral generally do not vanish. The integral can contribute a number of non-zero components to the vector, but as argued previously from \eqref{eqn:app_fn_deriv_three_terms} onwards, the number is finite. On the other hand, there are insufficient $z(1-z)$ factors in the remaining boundary term to ensure that every $c$-component contribution it provides, will vanish. From the polynomial form of $\tilde{P}_{a}$ and $\tilde{P}_{c}$, for any $a$, the vector $\cI_{c}^{a,2,1}$ can therefore potentially contain infinitely many non-zero components. If we further specify state $a=1$, we can observe such a vector:
\begin{equation}
\label{eqn:app_fn_lower_c_vec}
\cI^{1,2,1} =
\begin{pmatrix}
0\,, & 6\,, & 0\,, & 2\sqrt{21}\,, & 0\,, & 2\sqrt{33}\,, & 0\,, & 6\sqrt{5}\,, & \dots
\end{pmatrix}
\end{equation}
which continues indefinitely, owing to the infinite range of $c$.

For the upper octant example, choose $\alpha = 1, \beta = 1$, to give element
\begin{equation}
\label{eqn:app_fn_cent_insert_up}
(\tilde{\cB}_{2d})_{(1,1), (1,1)} = \sum_{c=0}^{\infty} \braket{\psi_{a}| S Z(1-Z) | \psi_{c}} \braket{\psi_{c}| Z(1-Z)S^{\dagger}| \psi_{a}} \,.
\end{equation}
Following the same process as above
\begin{align}
\cI^{a,1,1} &= \braket{\psi_{a} | S Z(1-Z) | \psi_{c}}
\\
&=i \sqrt{(2a+1)(2c+1)} \int_{0}^{1} \tilde{P}_{a}(z) \partial_{z} \left[ z(1-z) \tilde{P}_{c}(z) \right] dz
\\
&= i \sqrt{(2a+1)(2c+1)} \left[ \tilde{P}_{a}(z) z(1-z) \tilde{P}_{c}(z) \bigg|_{0}^{1} - \int_{0}^{1} \left(\partial_{z} \tilde{P}_{a}(z) \right) z(1-z) \tilde{P}_{c}(z) dz \right]
\end{align}
Clearly in this case the boundary term will vanish for all $a$ and $c$ while again, as shown from equation \eqref{eqn:app_fn_deriv_three_terms} onwards, the remaining integral will contribute a finite number of non-zero vector components for a chosen $a$. The resulting vector with state $a=1$ for example is
\begin{equation}
\label{eqn:app_fn_upper_oct_c_vec}
\cI^{1,1,1} =
\begin{pmatrix}
-\frac{i}{\sqrt{3}}\,, & 0\,, & \frac{i}{\sqrt{15}}\,, & 0\,, & 0\,, & 0\,, & 0\,, & 0\,, & \dots
\end{pmatrix}\,.
\end{equation}

%%%%%%%%%%%%%%%%%%%%%%%%%%%%%%%%%%%%%%%%%%%%%%%%%%%%%%%%%%%%%%%%%%%%%%%%%%%%%%%%%%
%%%%%%%%%%%%%%%%%%%%%%%%%%SUBSECTION%%%%%%%%%%%%%%%%%%%%%%%%%%%%%%%%%%%%%%%
%%%%%%%%%%%%%%%%%%%%%%%%%%%%%%%%%%%%%%%%%%%%%%%%%%%%%%%%%%%%%%%%%%%%%%%%%%%%%%%%%%

\section{Regularising matrix elements}
\label{app:fme_regularise}
The examples of Appendix \ref{app:fme_up_low_examples} showed how the vector calculations imply that the corresponding matrix elements can blow up, depending on octant choice. Here we discuss the regularisation of such matrix elements. In the lower octant case, the vector of \eqref{eqn:app_fn_lower_c_vec} leads to an infinite sum for the matrix element
\begin{align}
(\tilde{\cB}_{2d})_{(2,1), (2,1)} &=
\sum_{c=0}^{\infty}\cI_{c}^{1,2,1}(\cI_{c}^{1,2,1})^{*}
\\
&= 0 + 36 + 0 + 84 + 0 + 132 + 0 + 180 + \dots
\\
&= 6\sum_{n=1}^{\infty} (2n-1)\left(1+(-1)^{n}\right)
\end{align}
where $\cI_{c}^{(a,\alpha,\beta)}$ is as defined in \eqref{eqn:app_fn_two_vecs_indexed_by_c}. By inserting a regularisation factor $e^{-\epsilon n}$ into the summand and series expanding about $\epsilon = 0$, the sum can be regularised
\begin{equation}
\label{eqn:app_fn_low_oct_reg_sum_central}
(\tilde{\cB}_{2d})_{(2,1), (2,1)} = \frac{12}{\epsilon^2} - \frac{6}{\epsilon}+ 2 - 2 \epsilon + \frac{4 \epsilon^2}{5} + \cO(\epsilon^3) \implies (\tilde{\cB}_{2d}^{\text{reg}})_{(2,1), (2,1)} = 2 \,.
\end{equation}
This was achieved by removing the $\epsilon^{-1}, \epsilon^{-2}$ divergences and then taking $\epsilon \to 0$. 

The proof in Appendix \ref{app:fme_main_proof} and examples of Appendix \ref{app:fme_up_low_examples} considered a central insertion of the complete set of states. Trialling other positions of insertion shows that the upper octant $c$-vectors can also become infinite. However, it appears that while the upper octant regularised matrix elements are consistent with any state insertion position, the lower octant elements produce different/inconsistent results. For example, by inserting the complete set of states in the $a = 1, \alpha=1, \beta =1$ upper octant case, between $SZ^2$ and $(1-Z)^2 S^{\dagger}$ instead, we have\footnote{These calculations were achieved using the analytic solution integration.}
\begin{align}
(\tilde{\cB}_{2d})_{(1,1), (1,1)} &= \sum_{c=0}^{\infty} \braket{\psi_{1}| S Z^2 | \psi_{c}} \braket{\psi_{c}| (1-Z)^2 S^{\dagger}| \psi_{1}}
\\
&=\frac{1}{3}  -4 + \frac{196}{15} -21 + 27  -33 + 39  -45 + 51  -57 + 63 + \dots
\\
&= \frac{1}{3}  -4 + \frac{196}{15} + \sum_{n=1}^{\infty} (-1)^{n} (6n+15)
\end{align}
and upon regularising the sum using the $e^{-\epsilon n}$ factor as above, we obtain
\begin{align}
(\tilde{\cB}_{2d})_{(1,1), (1,1)}
&=\frac{1}{3}  -4 + \frac{196}{15}  -9 + \frac{15 \epsilon}{4} + \frac{3 \epsilon^2}{8} + \cO(\epsilon^3)
\\
(\tilde{\cB}_{2d}^{\text{reg}})_{(1,1), (1,1)} &= \frac{2}{5} \,.
\end{align}
where we have taken $\epsilon \to 0$, in the last line. This result is in agreement with the central insertion example \eqref{eqn:app_fn_cent_insert_up}, as seen by evaluating the inner product of $\cI^{(1,1,1)}$ from \eqref{eqn:app_fn_upper_oct_c_vec}: $\sum_{c=0}^{\infty} \cI_{c}^{1,1,1}(\cI_{c}^{1,1,1})^{*} = \frac{2}{5}$. Repeating a similar exercise for the lower octant example $a=1, \alpha = 2, \beta =1$, we trial an insertion of form
\begin{align}
(\tilde{\cB}_{2d})_{(2,1), (2,1)}  &= \sum_{c=0}^{\infty} \braket{\psi_{1}| S^2 Z^2 | \psi_{c}} \braket{\psi_{c}|(1-Z)^2 (S^2)^{\dagger}| \psi_{1}}
\\
&=  0 + 36 -540 + 3024 - 10800 + 29700 + \dots
\\
\label{eqn:app_fn_inf_sum_low_oct}
&=\sum_{n=1}^{\infty} (-1)^{n} (-3 n^2 + 12n^3 -15 n^4 + 6 n^5)
\\
&= -\frac{27 \epsilon}{8} + \frac{27 \epsilon^2}{16} + \cO(\epsilon^3) \,.
\end{align}
where the last line applies the series expansion around $\epsilon$, after inserting regularisation factor $e^{-\epsilon n}$ in \eqref{eqn:app_fn_inf_sum_low_oct}. Hence, upon taking $\epsilon\to 0$ we have
\begin{equation}
(\tilde{\cB}_{2d}^{\text{reg}})_{(2,1), (2,1)}  = 0\,,
\end{equation}
which contradicts the previous result in \eqref{eqn:app_fn_low_oct_reg_sum_central}. Such outcomes can naturally lead to inconsistent positivity calculations, which is why the positive upper octant was used for the 2d matrix calculations.

\section{Slack variable method}
\label{app:slack_var_method}
Here we provide an alternative method to search the $E$ space, referred to as the slack variable method. In the context of the quantum mechanical bootstrap, this approach was introduced in \cite{Berenstein:2022unr}, and here we provide a brief overview of the algorithm. Given the bootstrap matrices are Hermitian by construction, to satisfy $\cB \succeq 0$ it suffices that the minimal (smallest) eigenvalue of $\cB$ is positive. In the context of optimisation, the objective is to
\begin{equation}
\text{maximise } \lambda_{\text{min}}[\cB(E)] \,.
\end{equation}
Here $\cB$ is the $K \times K$ bootstrap matrix, and its eigenvalues $\lambda_{i}$ with $i = 1, \dots, K$, depend on initial data $E$. If the optimised/maximised minimal eigenvalue $\lambda_{\text{min}}$ is negative, the initial data $E$ is rejected. Equivalently, this optimisation problem can be phrased using a slack variable $t$
\begin{align}
\label{eqn:numres_slack_var_constraint}
\begin{split}
&\text{maximise } \qquad t \,,
\\
&\text{subject to } \qquad \cB(E)- t \mathbb{I} \succeq 0 \,.
\end{split}
\end{align}
This is a semidefinite programming problem in linear matrix inequality form, where the only initial/primal variable is $t$. In this description, at any given $E$ the method can always identify an optimal $t$ such that constraint \eqref{eqn:numres_slack_var_constraint} is satisfied. Since these optimal $t_{*}$ values depend continuously on $E$, we can use this dependence as an indicator of where the physical energies exist. Regions of $E$ for which $t_{*} \geq 0$ indicate $\cB \succeq 0$, and conversely, $t_{*} <0$ indicates $\cB \nsucceq 0$. This method is particularly useful for a large initial data space, but as will be shown, it still performs well for the single data  $E$ here, and supports the findings of the step search approach.

We apply the slack variable method to the $\cB_{1d}$, $\cB_{2d}''$ and $\tilde{\cB}_{2d}''$ matrices, and comment on their results. Each figure plots $\log{|t_{*}|}$ vs. $E$ for matrix sizes $K=2,4,6,8$, and  use step size $10^{-2}$ and energy range $E \in [0,50
]$ as before. The energy eigenvalues of the system are displayed in the figures as black dashed lines and the singularities of the $K=8$ bootstrap matrices are presented by red dashed lines. Taking the log of $t_{*}$ allows the significant behaviour to be seen more clearly, primarily the \textit{inverted spike} (using the name assigned in \cite{Berenstein:2022unr}) behaviour appearing around the eigenvalues.

The $\cB_{1d}$ plot is given in Figure \ref{fig:b1dslack}. A clear example of the inverted spike behaviour occurs for the $K=4,6,8$ curves, around $E=2$ (second black dashed line), where an arch is bound by two inverted spikes: the negative $\log{|t_{*}|}$ values with large magnitude. The width of the arch reduces with increasing $K$, i.e. the spikes become closer together. This is equivalent to a reduction in size of the band in the step search. The $K=2$ curve does not feature inverted spikes around $E=2$, since the local energy band could not be constrained at this matrix size. The final feature is the sharp spike at the singularities. $\cB_{1d}$ evaluated at energies close to these singularities, yield large eigenvalues. This implies that $|t_{*}|$ becomes (relatively) large in order to ensure $\cB_{1d} -t_{*} \mathbb{I} \succeq 0$.

\begin{figure}[t!]
\begin{center}
\includegraphics[width=15cm]{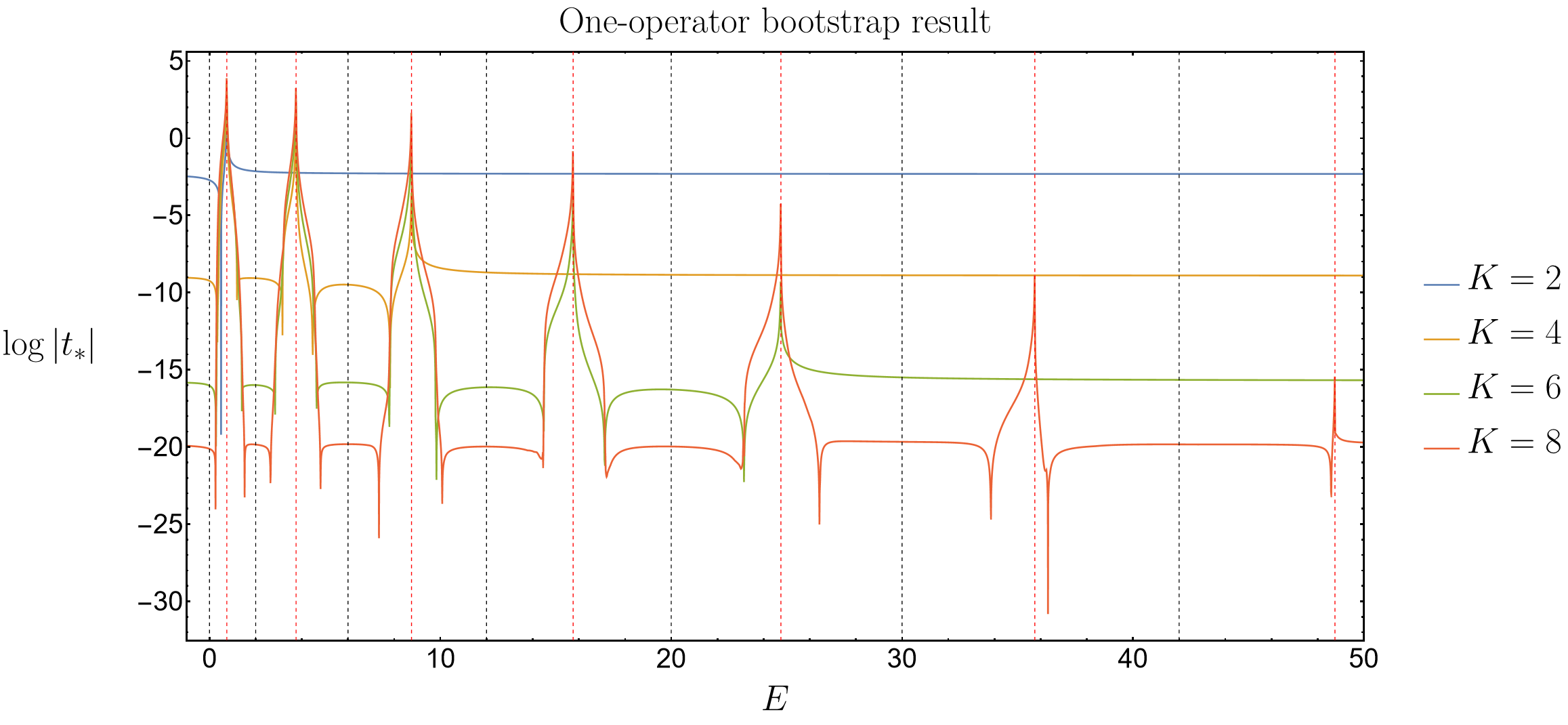}
\end{center}
\caption{Plotting the log of the optimal slack variable $t_{*}$ vs. energy $E$ for $\cB_{1d}$, at matrix sizes $K = 2, 4, 6, 8$. Pairs of inverted spikes appear around energy eigenvalues. At $E=2$ we see the distance between neighbouring spikes decrease from $K=4$ to $K=8$, corresponding to the shrinking of the energy band. The plot also indicates the results for small negative $E$.}
\label{fig:b1dslack}
\end{figure}

The $\cB_{2d}''$ plot of Figure \ref{fig:b2dpp_slack} exhibits different behaviour. The inverted spikes  predominantly appear close to the singularities. Only at $E=2$ do the spikes become somewhat more distinct; moving slightly away from these singularities. This agrees with Figure \ref{fig:b2dpp_step}, showing how the bootstrap performed poorly in constraining the energy into bands compared to the $\cB_{1d}$ case.
\begin{figure}[t!]
\begin{center}
\includegraphics[width=15cm]{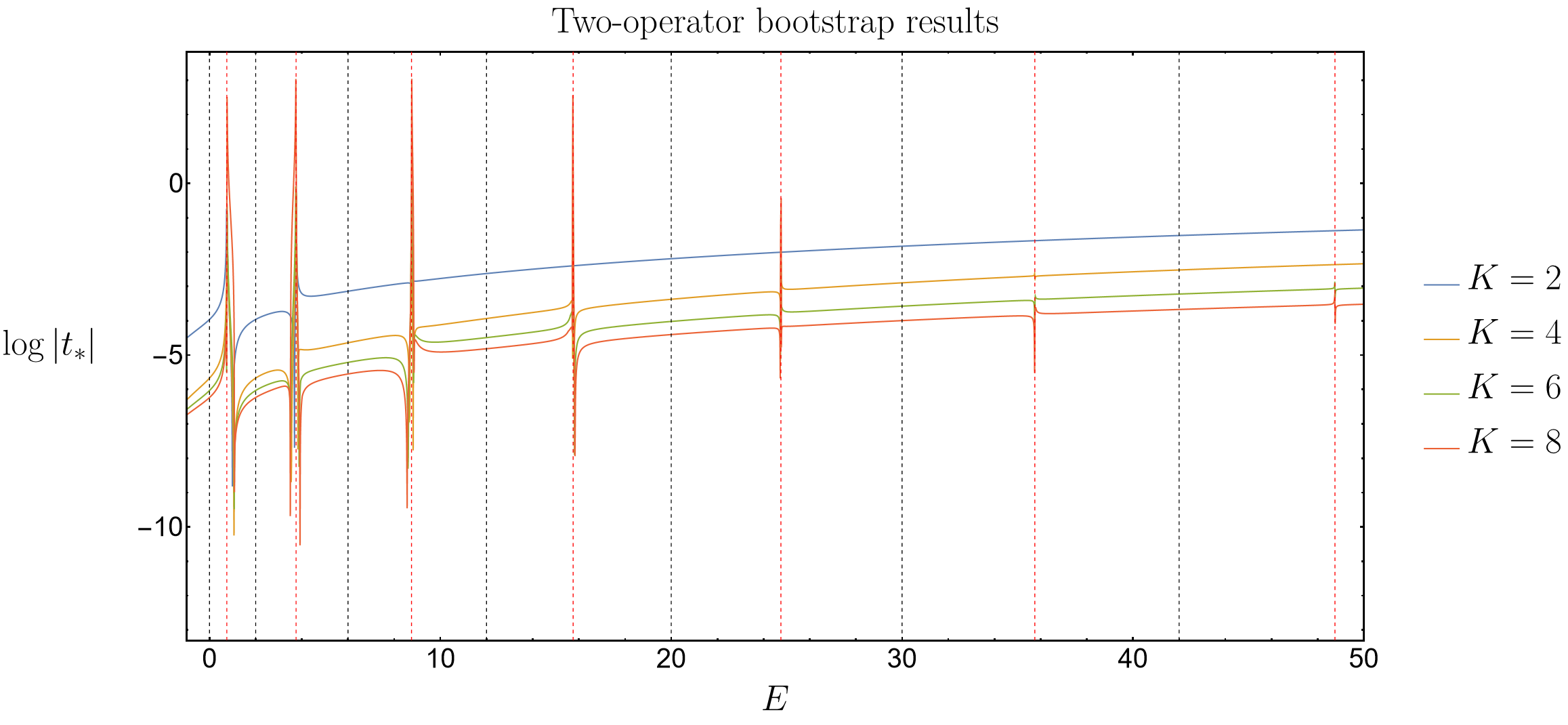}
\end{center}
\caption{$\log{|t_{*}|}$ vs. $E$ for $\cB_{2d}''$ at matrix sizes $K=2,4,6,8$. The inverted spike behaviour occurs close to the singularities (red dashed) and shows that two neighbouring spikes do not significantly approach each other as $K$ increases. This demonstrates the difficulty in constraining the energies for this matrix.}
\label{fig:b2dpp_slack}
\end{figure}

Finally, we consider the $\tilde{\cB}_{2d}''$ results in Figure \ref{fig:tilb2dpp_slack}. Here we see significantly negative $\log{|t_{*}|}$ values occurring at specific points in the $E$ space. This is where two spikes join and become indistinguishable (the arch between becomes point like), converging on a single energy eigenvalue. For example, the $K=6$ curve appears to show joined inverted spikes at the first five energy eigenvalues, followed by pairs of spikes enclosing $E=30, 42$ corresponding to energy bands at this $K$. On the other hand the $K=8$ curve finds all energy eigenvalues for $E \in [0,50]$. The fact that $\log{|t_{*}|}$ is negative with large magnitude at these joined-spike points is indicative of a $t=0$ crossing, implying $\tilde{\cB}_{2d}'' \succeq 0$ is effectively satisfied. 

\begin{figure}[t!]
\begin{center}
\includegraphics[width=15cm]{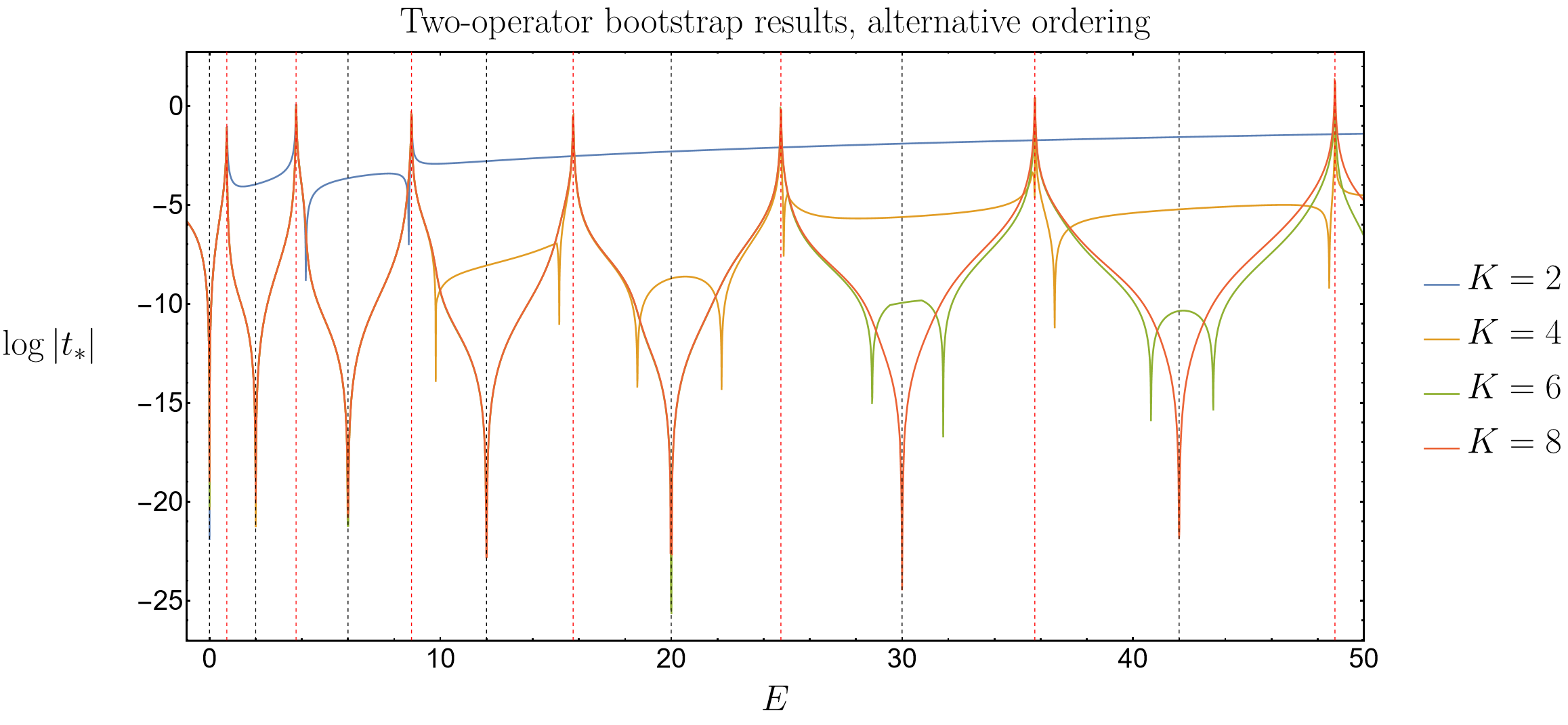}
\end{center}
\caption{$\log{|t_{*}|}$ vs. $E$ for $\tilde{\cB}_{2d}''$ at matrix sizes $K=2,4,6,8$. This plot features inverted spikes that surround exact energy eigenvalues e.g. $K=6$ at $E=30,42$, but also spikes that have joined together e.g. $K=8$ at $E=30,42$. The converging of two neighbouring spikes implies the location of an energy eigenvalue. The $K=8$ curve is able to find all seven energy eigenvalues in the range presented. }
\label{fig:tilb2dpp_slack}
\end{figure}

All numerics were generated in Mathematica \cite{Mathematica} and the slack variable calculations were obtained using the ``SemidefiniteOptimization" function, with method option ``CSDP". We minimised over $-t$, as opposed to the alternative convention of maximising over $t$. Since this function is limited to machine precision, to identify larger eigenvalues with a higher precision using the slack variable search, alternative programs should be considered. As a diagnostic tool however, it is sufficient for the present workings.

%%%%%%%%%%%%%%%%%%%%%%%%%%%%%%%%%%%%%%%%%%%%%%%%%%%%%%%%%%%%%%%%%%%%%%%%%%%%%%%%%%%%%%%%%%%%%SECTION%%%%%%%%%%%%%%%%%%%%%%%%%%%%%%%%%%%%%%%%%%%%%%%%%%%%%%%%%%%%%%%%%%%%%%%%%%%%%%%%%%%%%%%%%%%%%%%%%%%%%%%%%

%%%%%%%%%%%%%%%%%%%%%%%%%%%%%%%%%%%%%%%%%%%%%%%%%%%%%%%%%%%%%%%%%%%%END%%%%%%%%%%%%%%%%%%%%%%%%%%%%%%%%%%%%%%%%%%%%%%%%%%%%%%%%%%%%%%%

\printbibliography
\end{document}